\documentclass[aoas]{imsart}

\RequirePackage{amsthm,amsmath,amsfonts,amssymb}
\RequirePackage{natbib}
\RequirePackage[colorlinks,citecolor=blue,urlcolor=blue]{hyperref}
\RequirePackage{graphicx}

\usepackage{url}

\usepackage{amsmath,amsfonts,amssymb,amsthm}
\usepackage{booktabs} %For table
\usepackage{subcaption}
\usepackage{url}
\usepackage{bbm}
\usepackage{bm}
\usepackage{mwe}
\usepackage{paralist} % for compact item and enum
\usepackage{indentfirst} % to indent first paragraph in each section
\usepackage{xcolor}

\DeclareMathOperator*{\argmin}{arg\,min}

\renewcommand{\le}{\leqslant}

\renewcommand{\leq}{\leqslant}
\renewcommand{\geq}{\geqslant}

\renewcommand{\succeq}{\succcurlyeq}

\newcommand{\bbzOm}{\mathrm{bias}_{\beta_Z}^{\mathrm{(OM)}} }
\newcommand{\bbzME}{\mathrm{bias}_{\beta_Z}^{\mathrm{(ME)}} }
\newcommand{\bbxME}{\mathrm{bias}_{\beta_X}^{\mathrm{(ME)}} }
\newcommand{\bbME}{\mathrm{bias}_{\beta}^{\mathrm{(ME)}} }

\newcommand{\Ozone}{\mathrm{O}_3}
\newcommand{\CO}{\mathrm{CO}}
\newcommand{\NitrogenDi}{\mathrm{NO}_2}
\newcommand{\PMBig}{\mathrm{PM}_{10}}
\newcommand{\PMSmall}{\mathrm{PM}_{2.5}}
\newcommand{\SulfurDi}{\mathrm{SO}_{2}}
\newcommand{\PPC}{\text{Partial } \mathrm{PC}_{+}}

\newcommand{\PC}{\mathrm{PC}_{+}}

\newcommand{\e}{\mathbb{E}}
\newcommand{\var}{\mathrm{Var}}
\newcommand{\cov}{\mathrm{Cov}}
\newcommand{\corr}{\mathrm{Corr}}

\newcommand{\real}{\mathbb{R}}

\newcommand{\tran}{\mathsf{T}}
\newcommand{\indep}{\raisebox{0.05em}{\rotatebox[origin=c]{90}{$\models$}}}

\newcommand{\ce}{\mathcal{E}}

\newcommand{\cw}{\mathcal{W}}
\newcommand{\cx}{\mathcal{X}}
\newcommand{\cmw}{\mathcal{M}_W}
\newcommand{\cmx}{\mathcal{M}_X}
\newcommand{\cy}{\mathcal{Y}}
\newcommand{\cz}{\mathcal{Z}}

\newcommand{\tx}{W}

\makeatletter
\newcommand{\skipitems}[1]{\addtocounter{\@enumctr}{#1}}
\makeatother

\newtheorem{theorem}{Theorem}

\newtheorem{corollary}{Corollary}
\newtheorem{proposition}{Proposition}

\theoremstyle{remark}
\newtheorem{definition}{Definition}

\begin{document}

\title{ \sffamily Biases in estimates of air pollution impacts: the role of omitted variables and measurement errors}

\begin{frontmatter}

\begin{aug}
        
\author[A]{\fnms{Dan M.}~\snm{Kluger}\ead[label=e1]{kluger@stanford.edu}}
\author[D] {David B. Lobell\ead[label=e5]{dlobell@stanford.edu}}
\and
\author[A]{\fnms{Art B.}~\snm{Owen}\ead[label=e2]{owen@stanford.edu}}
\address[A]{ Department of Statistics, Stanford University    
\printead[presep={,\ }]{e1,e2}}
\address[D]{ Department of Earth System Science and Center on Food Security and the Environment, Stanford University \printead[presep={,\ }]{e5}}

\end{aug}

\begin{abstract}
Observational studies often use linear regression to assess the effect of ambient air pollution on outcomes of interest, such as human health indicators or crop yields. Yet pollution datasets are typically noisy and include only a subset of the potentially relevant pollutants, giving rise to both measurement error bias (MEB) and omitted variable bias (OVB). While it is well understood that these biases exist, less is understood about whether these biases tend to be positive or negative, even though it is sometimes falsely claimed that measurement error simply biases regression coefficient estimates towards zero. In this paper, we study the direction of these biases under the realistic assumptions that the concentrations of different types of air pollutants are positively correlated with each other and that each type of pollutant has a nonpositive association with the outcome variable. We demonstrate both theoretically and using simulations that under these two assumptions, the OVB will typically be negative and that more often than not the MEB for null pollutants or for pollutants that are perfectly measured will be negative. We also use a crop yield and air pollution dataset to show that these biases tend to be negative in the setting of our motivating application. We do this by introducing a validation scheme that does not require knowing the true coefficients. While this paper is motivated by studies assessing the effect of air pollutants on crop yields, the findings are also relevant to regression-based studies assessing the effect of air pollutants on human health outcomes. The validation scheme can also be used to empirically study OVB or MEB in other contexts. 

\end{abstract}


\end{frontmatter}

\section{Introduction} \label{sec:Intro}

Identifying and estimating the effect of ambient air pollutants on outcome variables such as crop yields or human health metrics is an important task from a policy perspective, as restrictions on emissions of air pollutants could lead to large improvements in these outcomes. To study the impact of air pollutants, researchers often must turn to observational approaches due to practical limitations in conducting randomized experiments under the conditions of interest. For example, due to the difficulty of controlling air pollutants in the open air environments where crops are typically found, the vast majority of experiments studying the effect of air pollutants on crop yields are conducted in greenhouses or chambers \citep{AinsworthNiceFACEversusChamberSidebar}. Meanwhile, due to ethical concerns, experiments assessing the effect of air pollutants on human health outcomes are limited to animal models or to trials where it is expected that the health consequences will be transient and reversible \citep{ControlledHumanInhalationEPA}.

Observational approaches to assess the causal effects of air pollutants on outcomes of interest also have major limitations. First, any claim about the causal effect of air pollutants using observational data relies either on an assumption that a particular variable is a valid instrument or on an assumption that there are no unmeasured confounders. Second, observational studies often rely on the assumption that the statistical models relating the variables are sufficiently well specified. Third, there is a sparse collection of ground truth measurements of air pollutant concentrations from air quality monitors, and in many locations where there are air quality monitors, only a fraction of the pollutants of interest are measured. In this paper, we focus on the consequences of the third issue, in settings where the investigator uses a linear regression model and where we assume the confounders are known, but the confounding pollutants are either sparsely measured or are measured with error.

In Figure \ref{fig:monitor_colocation}, we see that even in the United States, a country with a relatively large number of high quality air pollutant monitors, certain pollutants of interest are sparsely monitored, and even in locations where a pollutant of interest is monitored, the remaining pollutants of interest are often not monitored. If an investigator wishes to assess the effect of air pollutants on an outcome variable such as crop yield this leaves them with two natural choices. One option is to just include one air pollutant of interest in their model, so that sites where other air pollutants are not measured need not be dropped from the analysis. A second option is to use error-prone estimates of the air pollution concentrations based on satellite imagery or other sources and to treat these estimates as the true concentrations when fitting the statistical model. Indeed, these two options are often deployed in the pollutant and crop yield literature. For example, various recent linear-regression-based observational studies assessing the impact of air pollutants on crop yields either only include one air pollutant in all of their regression models \citep{McGrath_OzoneUS2015paper,LobellBurneyERL2021,Da2022} or relegate controlling for other pollutants to robustness checks \citep{ColumbiaOzone2019,YiEtAlOzonePM25MaizeChinaObs, LobellNOx2022}. Recent linear-regression-based studies assessing the association between air pollutants and crop yields have also used various types of error-prone estimates of pollutant concentrations such as inverse distance weighted monitor measurements \citep{ColumbiaOzone2019,YiEtAlOzonePM25MaizeChinaObs} satellite-based proxies \citep{LobellNOx2022,Da2022} satellite-based machine learning predictions \citep{HeEtAlOzoneMLpredsVsCropYieldsChina} or a data product based on chemical transport models, emissions inventories and meteorology observations \citep{Yi2016_OzoneChinaWheatObservational}.

\begin{figure}[t]
    \centering
    \includegraphics[width=0.95 \hsize]{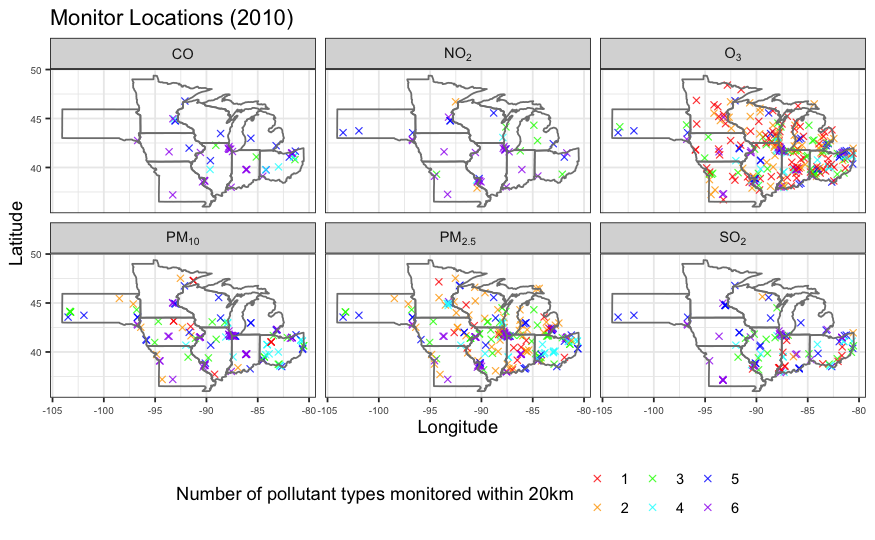}
    \caption{EPA air pollutant monitor locations in 2010 in the US Corn belt. Each panel presents the locations of air pollutant monitors for the air pollutant depicted in the panel title. The colors indicate how many of the 6 pollutants of interest ($\CO,\NitrogenDi,\Ozone,\PMBig,\PMSmall,$ and $\SulfurDi$) were also monitored within a 20 km radius of the plotted monitor.}
    \label{fig:monitor_colocation}
\end{figure}

Investigators should be cautioned that the first approach of only including one pollutant type in the model leads to omitted variable bias (OVB) while the second approach of plugging in error-prone proxies leads to a measurement error bias
%bias due to measurement error 
(MEB). OVB and MEB are well known types of bias and the formulas are well established in the literature; however, less is understood about the direction of these biases because they depend on unknown parameters. Particularly little can be said about the MEB in settings where multiple covariates have measurement error and the bias formula does not lead to usable results, as noted in 
some leading econometrics textbooks (see footnote 2 of \cite{AbelCMESeveralRegressors}, referencing textbooks \cite{WooldridgeText} and \cite{GreeneEconemtricsTextbook}). \cite{RobertsAndWhitedChapter} even state that ``little research on the direction and magnitude of these inconsistencies [limiting biases] exists because biases in this case are typically unclear and complicated to derive" (p. 504). 

In this paper, we introduce two assumptions and study the directions of the OVB and MEB under these assumptions. Our two assumptions are both realistic and particular to the setting of observational studies assessing the association between air pollutants and human health or agricultural outcomes. The first assumption is that each pollutant has either a negative association with the outcome variable or no association with the outcome variable. The second assumption is that the pairwise correlations between each pair of relevant pollutants are all positive. Under these two assumptions, we show that the OVB tends to be negative. We also demonstrate that MEB can easily be negative. In particular, under certain stated conditions and the classical measurement error model, we prove that pollutant coefficient estimates suffer from a combination of attenuation bias (shrinkage of estimates towards zero) and negative additive bias. As a corollary, under these conditions, the coefficient estimates for null pollutants as well as those for pollutants that are perfectly measured will be negatively biased as a result of the other, correlated, error-prone pollutants. When loosening modelling assumptions about the measurement error structure in simulations, our findings regarding negative MEB still tend to hold. Finally in a validation scheme that uses a pollutant and crop yield dataset from our motivating application, we find that the MEB for perfectly measured pollutants still tends to be negative. Notably, our data-based validation imposes no assumptions about the covariance matrix of the data or about the measurement error structure.

Our theoretical findings about the direction of OVB are likely not novel and are intuitive to econometricians (see, for example Section 3.3 of \cite{CinelliAndHazletOmVarSign}); however, our results concerning the direction of MEB are, to our knowledge, novel. 
The scenarios we demonstrate of negative MEB for null pollutants are particularly concerning given that it is sometimes falsely assumed that measurement error always results in attenuation bias (see Myth 2 in \cite{MythsEpi5Paper}, a recent paper enumerating common myths about measurement error). Therefore, researchers may not only report negative effects for pollutants which do not actually affect the outcome, they may very well falsely believe that the estimated effects are conservative (i.e., biased towards zero). %Such scenarios highlight a need for the development of methods and datasets that can be used to appropriately remove MEB in studies assessing the effects of ambient air pollutants.

In the air pollution and crop yield literature, the direction and magnitude of the MEB is not often discussed. In two cases where it has been discussed, only one pollutant was being studied and attenuation bias was assumed \citep{Yi2016_OzoneChinaWheatObservational,ColumbiaOzone2019}. Another contribution of this paper is that we study the MEB using a dataset and regression model that are representative of those commonly used in the pollution and crop yield literature. We consider models where there is one perfectly measured pollutant of primary interest and another pollutant that is controlled for using an error-prone proxy. We found that the MEB for the coefficient of the pollutant of primary interest tended to be negative for the ordered pairs of pollutants considered. In addition, after converting the outcome variable to the log scale and standardizing pollutant concentrations to be in units of the exposure limits recommend by the World Health Organization (WHO), this MEB was roughly negative 1--5\% on average.

The organization of this paper is as follows. In Section \ref{sec:Setup}, we introduce the mathematical setup and notation and present general formulas for the OVB and MEB, with the derivations of these formulas presented in Appendix \ref{sec:formula_deriv}. In Section \ref{sec:Theory}, we provide precise assumptions under which the OVB and the MEB are guaranteed to be negative, with proofs of theoretical guarantees in Appendix \ref{sec:ProofOfTheorems}. Because some of the assumptions needed for theoretical guarantees about the direction of the MEB may not hold in practice, in Section \ref{sec:BiasSims}, we use simulations to show that these theoretical findings still tend to hold under violations of such assumptions. Our simulations also demonstrate that as the pairwise correlations between pollutants increase, the chances that the OVB and the MEB are negative increase. In Section~\ref{sec:ValidationOnDataset}, we investigate OVB and MEB using some crop yield data, sparsely sampled air pollutant concentrations, and widely available proxies for the air pollutant concentrations.  We are able to study these biases despite not knowing the true regression coefficients in any of the models. To study OVB we compare regression coefficients for a pollutant of interest when some other measured pollutant is either retained or omitted. A similar strategy that compares using either a measured pollutant versus its proxy lets us study the MEB. Our empirical results are consistent with our theoretical results and demonstrate that the OVB and MEB of the perfectly measured pollutants tend to be negative. In Section \ref{sec:Discussion}, we discuss the limitations and implications of our findings, recent methods for removing MEB, as well as methodological and dataset developments that are needed to rigorously remove MEB in regression-based studies assessing the impacts of ambient air pollutants. 

\section{Setup, notation and general bias formulas} \label{sec:Setup}

In this section, we introduce notation and formal assumptions for the sampling and linear regression setup. We then present formulas for OVB and MEB, and discuss how the two types of bias relate to each other. Our bias formulas are asymptotic, so throughout the text, we use the word `bias' to describe estimation error that does not disappear in the infinite sample size limit, rather than the typical definition involving expected values, although the two definitions coincide asymptotically.

We start by defining the random variables, covariance matrices, and estimands of interest. Let $Y$ denote the outcome variable, let $Z = \big(Z^{(1)},\dots,Z^{(p)} \big)  \in \mathbb{R}^p$ denote the vector of $p$ covariates which are measured without error, and let $X= \big(X^{(1)},\dots,X^{(d)} \big) \in \mathbb{R}^d$ denote the actual value of the $d$ covariates which are not precisely measured by their available estimates $\tx \in \mathbb{R}^d$. It is convenient to define the covariance matrix of the random vector $(Z,X,\tx)$ in the following block-wise manner,

$$ \cov \Bigg( \begin{bmatrix} Z \\ X  \\ \tx \end{bmatrix}  \Bigg) \equiv \begin{bmatrix} A & B & C \\ B^\tran & D & F  \\ C^\tran & F^\tran & G \end{bmatrix},  \text{  where  } A \in \mathbb{R}^{p \times p}, \ B,C \in \mathbb{R}^{p \times d},\text{ and } D,F,G \in \mathbb{R}^{d \times d}.$$ It also helps to define the covariance matrix of the available covariates $(Z,\tx)$ as well as the covariance matrix of $(Z,\tx)$ and its error as an approximation of $(Z,X)$ with 
\begin{align}\label{eq:thosecovs}
\begin{split}
\Sigma_M &\equiv \cov \Big( \begin{bmatrix} Z \\ \tx \end{bmatrix} \Big) =\begin{bmatrix} A & C \\ C^\tran & G \end{bmatrix} \quad \text{and} \quad \\ 
\Sigma_{M,E} &\equiv \cov \Big( \begin{bmatrix} Z \\ \tx \end{bmatrix} , \begin{bmatrix} Z-Z \\ X-\tx \end{bmatrix} \Big) =\begin{bmatrix} 0 & B-C \\ 0 & F^ \tran -G \end{bmatrix}.
\end{split}
\end{align} The estimands of interest are the following population regression coefficients \begin{equation}\label{eq:def_reg_coef}
    (\alpha,\beta_Z,\beta_X) \equiv \argmin_{a \in \mathbb{R}, b_Z \in \mathbb{R}^{p},b_X \in \mathbb{R}^{d}} \e \Big[ \frac{1}{2}\big(Y-a-Z^\tran b_Z -X^\tran b_X \big)^2 \Big], \quad \beta \equiv (\beta_Z,\beta_X).
\end{equation}

Suppose there are $N$ data samples with $(Y_i,Z_i,X_i,\tx_i) \in \mathbb{R}^{1+p+d+d}$ denoting the $i$th data sample for $i=1,\dots,N$. To study OVB and MEB, we assume these regularity conditions:

\begin{enumerate}[(i)]
    \item The $N$ samples $\big(Y_i,Z_i,X_i,\tx_i \big)_{i=1}^N$ are identically distributed and the first and second sample moments of $\big(Y_i,Z_i,X_i,\tx_i \big)_{i=1}^N$ converge in probability to the corresponding first and second population moments of $(Y,Z,X,\tx)$ as $N \to \infty$. Here, the second population moments of a random vector $V$, include all entries in $\e[VV^\tran]$, not just the entries on the main diagonal.
    
    \item Defining $\varepsilon_i \equiv Y_i- \alpha -Z_i^\tran \beta_Z -X_i^\tran \beta_X$ for each $i=1,\dots,N$,  $\varepsilon_i$ is uncorrelated with the measurement error $W_i-X_i$.
       
    \item The covariance matrix of $(Z,W)$ and the covariance matrix of $(Z,X)$ are both positive definite.
\end{enumerate}

We use Condition (i) rather than the stronger IID assumption typically seen in the measurement error literature, because we expect the IID assumption to not hold in our motivating applications due to be spatial correlations between the samples. Condition (ii) that assumes $\varepsilon_i$ is uncorrelated with the measurement error $W_i-X_i$ is common in the measurement error literature in order to obtain tractable results, although the precise assumption and assumption name vary in different texts. \cite{FongAndTyler2021} refer to the equivalent assumption that $\varepsilon_i$ is uncorrelated with both $\tx_i$ and $Z_i$ as the exclusion restriction assumption (the exclusion restriction assumption is equivalent to Condition (ii) because by definition of $\varepsilon_i$ and by~\eqref{eq:def_reg_coef}, $\varepsilon_i$ must be uncorrelated with $Z_i$ and $X_i$). The textbook \cite{CarrolStefanski06} often uses the related assumption that $W \indep Y |X,Z$, and refer to this assumption as the nondifferential measurement error assumption. Condition (iii) is likely to hold and is a standard assumption used when running a regression of $Y$ on $(Z,W)$ and on $(Z,X)$. Condition (iii) ensures that all covariance matrices that need to be inverted throughout the text are indeed invertible.

An investigator wishing to regress $Y$ on $(Z,X)$ to estimate $\beta_Z$ and $\beta_X$ will be impeded by the fact that they have few or no measurements of $X$. They may opt to simply regress $Y$ on $Z$ using OLS and not provide an estimate for $\beta_X$. Let $\hat{\beta}_Z^{\text{(OM)}}$ denote their estimator for $\beta_Z$. In Appendix~\ref{sec:ovbderivation} we show this estimator has bias given by
\begin{equation}\label{eq:OmVarBias}
    \bbzOm  \equiv \text{plim}_{N \to \infty} \big[ \hat{\beta}_Z^{\text{(OM)}} - \beta_Z \big] = A^{-1} B \beta_X.
\end{equation} 
Alternatively, the investigator may opt to regress $Y$ on $(Z,W)$ using OLS, and estimate $\beta_X$ with the estimated regression coefficients for $\tx$. In this case, their estimate of $\beta  = (\beta_Z, \beta_X)$, call it $\hat{\beta}^{\text{(ME)}} = (\hat{\beta}_Z^{\text{(ME)}},\hat{\beta}_X^{\text{(ME)}})$ will have asymptotic bias given by,
\begin{equation}\label{eq:MEbiasAll}
    \bbME \equiv \text{plim}_{N \to \infty} \big[ \hat{\beta}^{\text{(ME)}} - \beta \big] = \Sigma_M^{-1} \Sigma_{M,E} \beta,
\end{equation} as shown in Appendix~\ref{sec:mebderivation}.
That appendix also shows that the asymptotic bias in $\hat{\beta}_Z^{\text{(ME)}}$, their estimate of $\beta_Z$, will be
\begin{equation}\label{eq:MEbiasZ}
    \bbzME \equiv \text{plim}_{N \to \infty} \big[ \hat{\beta}_Z^{\text{(ME)}} - \beta_Z \big] = \big( A- C G^{-1} C^ \tran \big)^{-1} (B - C G^{-1} F^ \tran ) \beta_X.
\end{equation}   

Because researchers may be more accustomed to reasoning about OVB than MEB, and because a number of studies in the pollutant and crop yield literature skirt away from measurement error issues by omitting error-prone pollutants from the regression, it is natural to ask whether OVB is a strictly worse issue than MEB. In Appendix \ref{sec:OmBiasVsMEbias}, we use formulas \eqref{eq:OmVarBias} and \eqref{eq:MEbiasZ} to investigate  whether the OVB is worse than the MEB. In summary, the results of Appendix \ref{sec:OmBiasVsMEbias} suggest that the entries of $\bbzOm$ will tend to be further away from $0$ than the corresponding entries of $\bbzME$, but this is not always the case, even under strong assumptions about the measurement error structure. The results for the cases we considered are summarized in Table \ref{table:OMbiasVsMEbias}, with subsequent subsections in Appendix \ref{sec:OmBiasVsMEbias} providing proofs of the results presented in the table. Therefore, while we believe including error-prone covariates in an OLS regression will, more often than not, reduce bias compared to omitting those covariates from the regression altogether, we also show that there are exceptions.

\section{Theoretical investigation of direction of biases}\label{sec:Theory}

In our motivating application, $Y$ is the crop yield, $Z$ is a vector of the covariates that are assumed to be measured without error, which includes weather covariates and covariates controlling for space and time. $X$ is the vector of actual pollutant concentrations, and $\tx$ is the vector of estimated pollutant concentrations. Pollutants that are perfectly measured could be included in the vector $Z$ or in the vector $X$, but not in both. When studying OVB, we choose to include any perfectly measured pollutants in $Z$. 

There are two assumptions that are reasonable to make in the setting of studying the association of air pollutant concentrations with crop yields: 
\begin{enumerate}[(I)]
    \item \textit{No Benefit Assumption:} The regression coefficients corresponding to the pollutants are nonpositive. % Could generalize to only requiring the error-prone pollutants to have negative regression coefficients
    \item \textit{Positive, Pairwise Pollutant Correlations (Pairwise $\PC$) Assumption:} The concentrations of the various pollutants are positively correlated with each other.
\end{enumerate} The No Benefit Assumption cannot be checked empirically but it is reasonable to assume that each air pollutant type does not have a net benefit on crop yields. The Pairwise $\PC$ Assumption can be checked empirically. For example, in Figure \ref{fig:PairwiseCorr}, we find that this assumption holds in the Midwestern United States.%, a region where nearly 30\% of the global corn and soybean crop production occurs \citep{USDAWASDE}.

\begin{figure}[!t]
    \centering
    \includegraphics[width=0.95 \hsize]{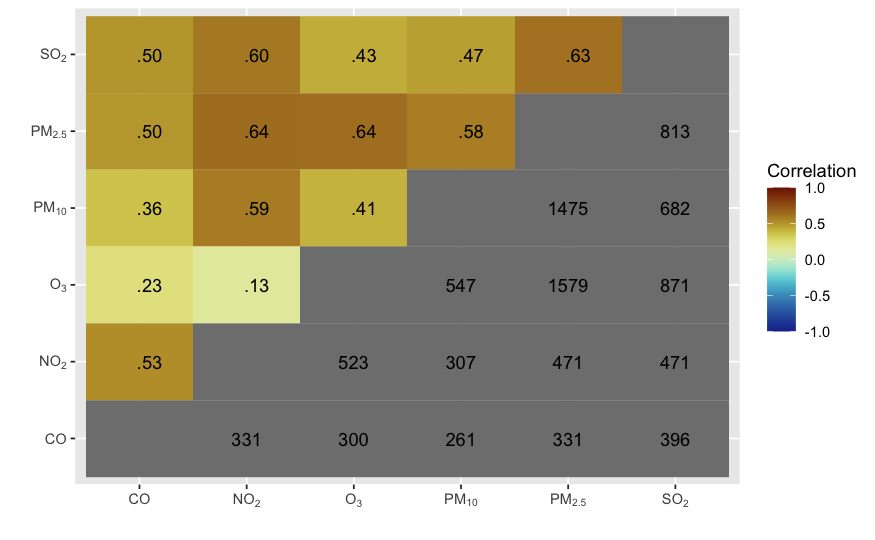}
    \caption{Pairwise correlations between pollutant concentrations and sample sizes used to compute them. The numbers in the top left are the sample pairwise correlations between each pair of pollutants. Pollutant concentration data was taken from the EPA's AQS monitors from the years 2002--2019 and the daily concentration data was transformed to yearly growing season averages (June--August). The numbers on the bottom right give the number of unique site-year pairs that were used to calculate the sample correlations. Many site-year pairs did not have measurements for all six pollutants, so missing values were dropped when computing each pairwise correlation.}

    \label{fig:PairwiseCorr}
\end{figure}

We will show that under these two assumptions,  the following four phenomena 
\textit{tend} to occur:
\begin{itemize}
    \item If the investigator does not include error-prone pollutants in the regression model, but includes a single pollutant that is perfectly measured, then

    \begin{enumerate}[\quad 1)]
        \item the OVB for the coefficient of the perfectly measured pollutant is negative.
    \end{enumerate}

    \item  If instead the investigator includes error-prone pollutants in the regression model, then
     \begin{enumerate}[\quad 1)]
        \skipitems{1}
        \item  pollutants that are perfectly measured have negative MEB,
        \item  pollutants that are null and have no association with crop yield have negative MEB,
        \item  and coefficient estimates for nonull, error-prone pollutants are subject to a combination of attenuation bias and negative additive bias.
\end{enumerate}
\end{itemize}

While phenomena 1--4 do not always occur under Assumptions (I) and (II), we show that they each \textit{tend} to occur, using theoretical arguments in Sections \ref{sec:OmittedVarBiasGeneralConditions} and \ref{sec:MEbias_genConditions}, using simulations in Section \ref{sec:BiasSims}, and using a real dataset in Section \ref{sec:ValidationOnDataset} (phenomena 1 and 2 only). % We also give related conditions under which each of the first four phenomena \textit{must} occur.
We are also interested in and study a fifth phenomenon that tends to happen
\begin{enumerate}[\quad1)]
\skipitems{4}
\item for the coefficient estimates of perfectly measured pollutants, the OVB is worse than the MEB in having
a larger absolute value.
\end{enumerate}
In Appendix~\ref{sec:OmBiasVsMEbias} we give some conditions where phenomenon 5 must hold and some conditions where it need not hold. We also see that it tends to hold in our simulations (Section~\ref{sec:BiasSims}).

Our theoretical results for the OVB and MEB, respectively, require the following modifications of the Pairwise $\PC$ Assumption:

\begin{enumerate}[(I')]
 \skipitems{1}
    \item \textit{Weak Positive Pollutant Partial Correlation (Weak 
 $\PPC$) Assumption:} For each $k \in \{1,\ldots,p \}$ such that $Z^{(k)}$ is a pollutant, and for each $j \in \{1,\ldots,d \}$, the partial correlation of the pollutants $Z^{(k)}$ and $X^{(j)}$ conditional on $Z  \setminus \{ Z^{(k)} \}$ (i.e., conditional on all measurement error free controls in the model besides $Z^{(k)}$) is positive.
\end{enumerate}

\begin{enumerate}[(I*)]
 \skipitems{1}
    \item \textit{Positive, Pairwise, Pollutant Partial Correlations (Pairwise $\PPC$) Assumption:} For each pair of pollutants, the  partial correlation between those two pollutants, conditional on all remaining entries of $(Z,X)$, is positive. 
\end{enumerate} Heuristically, in the simplified setting of only one perfectly measured, non-omitted pollutant, the Weak $\PPC$ Assumption supposes that the concentration of that pollutant has positive associations with the concentrations of each of the omitted pollutants when controlling for all of the non-pollutant covariates in the model (such as weather and space-time controls). The Pairwise $\PPC$ Assumption requires that for each pair of pollutants, the two pollutants have concentrations that are positively associated with eachother when controlling for all other covariates in the model (including concentrations of other pollutants). This assumption is generally more stringent than the Pairwise $\PC$ assumption because even when two variables are clearly positively correlated, the positive association between the variables may not hold when controlling for other covariates.

To prove theoretical results about the MEB, we must also make assumptions about the measurement error structure. For mathematical tractability, we use the following assumption

\begin{enumerate}[(I)]
 \skipitems{2}
    \item \textit{Classical, Uncorrelated Measurement Error (CUME) Assumption:} The measurement error is classical and is uncorrelated across the pollutants. That is $W=X+E$ with $E \indep (Z,X,Y)$, and $\Sigma_E \equiv \cov(E,E)=\text{diag}(a_1,\dots,a_d)$ where $a_j \geq 0$ for $j=1,\dots,d$.
\end{enumerate} The classical measurement error model in this assumption is common in the measurement error literature \citep{CarrolStefanski06}; however, the assumption that the measurement errors are uncorrelated across pollutants is more stringent than the classical measurement error assumption and is needed for the proof of our theoretical results about the MEB. Heuristically, the CUME Assumption states that the measurement error is simply additive noise that is uncorrelated with any of the variables of interest and that the measurement error for one pollutant is uncorrelated with that for the other pollutants.

To demonstrate that phenomena 1--4 each tend to occur under the No Benefit and Pairwise $\PC$ Assumptions, in Section \ref{sec:OmittedVarBiasGeneralConditions} we prove phenomenon 1 \textit{must} occur under the No Benefit Assumption and the related Weak $\PPC$ Assumptions. Then in Section \ref{sec:MEbias_genConditions} we prove that under the No Benefit, Pairwise $\PPC$, and CUME Assumptions, phenomena 2--4 \textit{must} occur. In Section \ref{sec:BiasSims}, we use simulations to demonstrate that phenomena 1--4 tend to hold under the No Benefit and Pairwise $\PC$ Assumptions for a more general class of correlation structures and measurement error models than those considered in Sections \ref{sec:OmittedVarBiasGeneralConditions} and \ref{sec:MEbias_genConditions}. Finally, in Section \ref{sec:ValidationOnDataset}, we demonstrate that phenomena 1 and 2 tend to occur using a real dataset from our motivating application without imposing any assumptions about the correlation structures or measurement error in the dataset.

\subsection{A theoretical guarantee for negative OVB}\label{sec:OmittedVarBiasGeneralConditions}

Here, we suppose that a pollutant is measured without error and is included in the covariate vector $Z$, as the $k$th entry in $Z=(Z^{(1)},\ldots,Z^{(p)})$, and the vector $X$ contains the error-prone pollutants which are to be omitted from the regression. In this setting we have the following proposition.

\begin{proposition}\label{prop:OmBiasNegConds}
Suppose Regularity Conditions (i)-(iii) hold. Under the No Benefit and Weak $\PPC$ Assumptions, $[\emph{bias}^{\emph{(OM)}}_{\beta_Z}]_k \leq 0$.
\end{proposition}

\begin{proof}
See Appendix \ref{sec:ProofOfTheorems} .  
\end{proof}

\subsection{A theoretical guarantee for negative MEB}\label{sec:MEbias_genConditions}

In this subsection, we suppose without loss of generality that 
%none of the 
no pollutant measurements are
contained in the vector $Z$, and that the vector $X$ contains only pollutants (and if the $j$th pollutant has no measurement error $W^{(j)}=X^{(j)}$). 
We will prove the following theorem, which will then be used to show that phenomena 2--4 are guaranteed to occur under the No Benefit, Pairwise $\PPC$, and CUME Assumptions.

\begin{theorem}\label{thoerem:PollutantBiasOmegaSigns} Suppose Regularity Conditions (i)-(iii) hold and that the Pairwise $\PPC$ and CUME Assumptions also hold. Then, \begin{equation}\label{eq:MEbiasAssumption}
      \emph{\text{bias}}^{\emph{(ME)}}_{\beta_X} = - \Omega \Sigma_E \beta_X \quad \text{ where }  \quad \Omega \equiv \big( \Sigma_E + D -C^ \tran A^{-1} C \big)^{-1}.
\end{equation}
    Furthermore, for all $j,j' \in \{1,\ldots,d \}$,  $\Omega_{jj} \in (0,1/a_j]$ and $\Omega_{jj'} \leq 0$ if $j \neq j'$.

\end{theorem}

\begin{proof} 

See Appendix \ref{sec:ProofOfTheorems}. \qedhere

\end{proof}

 An interesting implication of Theorem \ref{thoerem:PollutantBiasOmegaSigns} is that under the No Benefit, Pairwise $\PPC$, and CUME Assumptions, the MEB for a particular pollutant can be written as the sum of an attenuation bias term and an additive nonpositive bias term.  In particular, recall that $\Sigma_E = \text{diag}(a_1,\ldots,a_d)$. By applying Theorem \ref{thoerem:PollutantBiasOmegaSigns}, it follows that in the setting of Theorem \ref{thoerem:PollutantBiasOmegaSigns} and under the No Benefit assumption, 
 \begin{equation}\label{eq:MEBiasDecomp}
    \Big[ \bbxME \Big]_j =  -\Omega_{jj} a_j \beta_{X,j} +
   \sum_{ \substack{  j' \in \{1,\ldots,d\} \\ j' \neq j} }- \Omega_{jj'} a_{j'} \beta_{X,j'} \quad \text{for } j=1,\ldots,d,
\end{equation} 
where the first term is an attenuation bias term, while second term is nonpositive. To see that the first term is an attenuation bias term, note that by Theorem \ref{thoerem:PollutantBiasOmegaSigns}, $\Omega_{jj} \in (0, 1/a_j]$, implying $0  \leq 1-\Omega_{jj} a_j < 1$. Further observe that if you take an unbiased estimator of $\beta_{X,j}$ and shrink it towards zero by a factor of $(1-\Omega_{jj} a_j)$, that shrunken estimator of $\beta_{X,j}$ will have bias $-\Omega_{jj} a_j \beta_{X,j}$. Hence the first term in \eqref{eq:MEBiasDecomp} can be thought of as an attenuation bias term. To see that the second term in \eqref{eq:MEBiasDecomp} is nonpositive, recall that $a_{j'} \geq 0$, $\beta_{X,j'} \leq 0$ (by the No Benefit Assumption), and $\Omega_{jj'} \leq 0$ for $j' \neq j$ (by Theorem \ref{thoerem:PollutantBiasOmegaSigns}).

The decomposition of the MEB into an attenuation bias term and a strictly nonpositive term in \eqref{eq:MEBiasDecomp} leads to two notable corollaries. 

\begin{corollary}\label{cor:NegBiasNoME}
In the setting of Theorem \ref{thoerem:PollutantBiasOmegaSigns} and under the No Benefit Assumption, if the $j$th entry of $X$ is perfectly measured, then $\big[ \emph{\text{bias}}_{\beta_X}^{\emph{\text{(ME)}}} \big]_j \leq 0$.
\end{corollary}

\begin{proof}
    If $X^{(j)}$ does not have measurement error, then $a_j=0$. Thus in Equation \eqref{eq:MEBiasDecomp}, the first term is zero, while as argued previously, the second term is nonpositive.
\end{proof}

An implication of this corollary is that the regression coefficients for perfectly measured pollutants could easily be downward biased due to the measurement error associated with other, correlated pollutants in the model.

\begin{corollary}\label{cor:NegBiasNulPol}
In the setting of Theorem \ref{thoerem:PollutantBiasOmegaSigns} and under the No Benefit Assumption, if the regression coefficient corresponding to the $j$th entry of $X$ is zero, then $\big[ \emph{\text{bias}}_{\beta_X}^{\emph{\text{(ME)}}} \big]_j \leq 0$.
\end{corollary}

\begin{proof}
    If the regression coefficient corresponding to $X^{(j)}$ is zero, $\beta_{X,j}=0$. Thus in Equation \eqref{eq:MEBiasDecomp}, the first term is zero, while as argued earlier, the second term is nonpositive.
\end{proof}

A notable implication of Corollary \ref{cor:NegBiasNulPol} is that, under the No Benefit, Pairwise $\PPC$, and CUME Assumptions, for null pollutants that have no association with the outcome variable, the bias due to measurement error must be negative or zero.

\section{Simulations}\label{sec:BiasSims}

While we have demonstrated phenomenon 1 always holds under the No Benefit and Weak $\PPC$ Assumptions and that phenomena 2--4 always hold under the No Benefit, Pairwise $\PPC$, and Classical, Uncorrelated Measurement Error (CUME) Assumptions, some of the assumptions may fail to hold in our motivating applications. In particular, the assumption that the measurement errors are uncorrelated across pollutants could easily fail to hold. For example, if the error-prone pollutant estimates $W^{(1)},\ldots,W^{(d)}$ are predictions from machine learning models that each use the same input features, we expect the measurement errors to be correlated across pollutants.
Further the Pairwise $\PPC$ Assumption requiring nonnegative partial correlations between each pair of pollutants conditional on all other variables in the model may not hold. For example, in certain settings declines in $\PMSmall$ or nitrogen oxides can lead to increases in $\Ozone$ \citep{PartialCorrNegOzoneExample1,PartialCorrNegOzoneExample2}.

In light of some of the assumptions not holding up in our motivating applications, we conduct one million Monte Carlo simulations to explore whether phenomena 1--4 still tend to hold outside the conditions where we have seen that they must hold. In each simulation we supposed $p=5$ covariates were not error-prone, and $d=5$ covariates were error-prone. We supposed that the last entry of $Z$ was a pollutant that was perfectly measured.

In each simulation we first generated a random covariance matrix of $(Z,X,W)$, in which the covariance matrix of $(Z,X)$ will not necessarily satisfy assumptions from Section~\ref{sec:Theory} about the partial correlations between pollutants, and in which the measurement error $W-X$ is not of the classical type and is correlated across pollutants. Details and justification for how we randomly generated covariance matrices of $(Z,X,W)$ are provided in Appendix \ref{sec:MCSimsRandomCovMatricesDetails}.

The simulated regression coefficients satisfied the No Benefit Assumption. % still held. 
In particular, we set $\beta_{X,4}=\beta_{X,5}=0$ so that the final two pollutants would be null, and the bias for null pollutants could easily be studied. We generated negative coefficients corresponding to the other pollutants (i.e., $\beta_{Z,5},\beta_{X,1},\beta_{X,2},$ and $\beta_{X,3}$).
Their absolute values were IID draws from %the distribution $\text{Gamma}(1.4,1.6)$, and subsequently multiplying those draws by negative 1. The shape and rate parameters for this 
a Gamma distribution chosen to fit 
%by applying the method of moments to the 
the $16$ coefficient estimates in Figure 2 of \cite{LobellBurneyERL2021}. 
Using the method of moments, the Gamma shape parameter was $1.4$
and the rate parameter was $1.6$.

After generating the covariance matrix of $(Z,X,W)$ and the entries of $\beta$ corresponding to the pollutants, we computed the OVB and MEB from formulas \eqref{eq:OmVarBias} and \eqref{eq:MEbiasAll}. 
We checked whether phenemona 1--5 occurred by checking whether $[\bbzOm]_5 \leq 0$, $[\bbME]_5 \leq~0$, $[\bbME]_{10}~\leq 0$, $[\bbME]_{6} \leq 0$ and $|[\bbzOm]_5|>|[\bbME]_5|$, respectively. 
%For phenomenon 5
%In order to assess the conclusions from Appendix \ref{sec:OmBiasVsMEbias} that OVB tends to be worse than MEB, 
%we checked whether $|[\bbzOm]_5|>|[\bbME]_5|$.
%the absolute value of $[\bbzOm]_5$ was larger than the absolute value of $[\bbME]_5$. 
To see whether these phenomena are more or less likely under certain correlation structures of $(Z,X)$, we used the generated covariance matrix of $(Z,X,W)$ to calculate the average pairwise correlation between the pollutants and to check whether the assumptions from Section \ref{sec:Theory} on the correlations and partial correlations between pollutants held.

We repeated this process $10^6$ times with results summarized in Table \ref{table:BiasSignSims} and Figure \ref{fig:BiasSignSims}. Table \ref{table:BiasSignSims} shows that among the simulations in which the pairwise correlations between pollutants are all positive, the OVB and the MEB for the perfectly measured pollutant both tend to be negative, while for pollutants measured with error, the MEB only tends to be negative for null pollutants. The OVB was observed to be negative more frequently than the MEBs were. When Weak $\PPC$ was met, we observed the phenomena to occur with similar frequency except notably, the OVB was always nonpositive as guaranteed by Proposition \ref{prop:OmBiasNegConds}. We also see that the OVB tended to be worse than the MEB, corroborating the findings in Appendix \ref{sec:OmBiasVsMEbias}.

\begin{table}[!t]
\caption{ 
Simulated frequencies of phenomena 1--5.
Rows 1--5 show frequencies of $[\bbzOm]_5 \leq 0$, $[\bbME]_5 \leq 0$, $[\bbME]_{10} \leq 0$, $[\bbME]_{6} \leq 0$ and $\big|[\bbzOm]_5 \big| > \big| [\bbME]_5  \big|$, respectively. %The final row shows the percentage of simulations in which the OVB was worse than the MEB for the perfectly measured pollutant in the sense that 
The first column is for all $1{,}000{,}000$ simulations.
%indicates the percentage of time each phenomenon occurred out of all of the one million simulations. 
The second column is for %indicates the percentage of time each phenomenon occurred out of the 
$63{,}178$ simulations in which the Pairwise $\PC$ Assumption held. The third column %indicates the percentage of time each phenomenon occurred out of 
is for the $220{,}119$ simulations in which the Weak $\PPC$ Assumption held.
All standard errors for the percentages displayed below are $<0.2\%$. We do not include a column for which the Pairwise $\PPC$ Assumption is met, because it was only met in 4 out of the 1 million Monte Carlo simulations that were run. 
}
\label{table:BiasSignSims}
\centering

\begin{tabular}{c l c c c } 
\toprule
  Phenomenon & Description & All sim.s  & Pairwise $\PC$  & Weak  $\PPC$ \\ 
\midrule
1&  OVB $<0$ & 73.9\%  & 93.4\% & 100.0\% \\ [1ex] 
\midrule
2& Non-error-prone pollutant MEB $<0$  & 66.3\% & 78.3\% & 88.0\% \\ [1ex] 
\midrule
3&Null pollutant MEB $<0$ & 61.3\% & 68.2\% & 61.1\% \\ [1ex] 
\midrule
$ \ \ 4^\dag$ &Nonnull pollutant MEB $<0$ & 17.8\% & 20.9\% & 17.9\% \\ [1ex] 
\midrule
5&$\big|\mathrm{OVB}\big|>\big|\mathrm{MEB}\big|$ & 79.9\% & 91.5\% & 95.1\%  \\ [1ex] 
 \bottomrule
\end{tabular}
\smallskip

\dag We remark that this row does not precisely correspond to Phenomenon 4 that the coefficient estimates for error-prone, nonnull pollutants are subject to a combination of attenuation bias and negative additive bias. Instead this row corresponds to an aspect of Phenomenon 4, and gives the proportion of simulations where the attenuation bias is smaller in magnitude than the negative additive bias. %\newline
\end{table}

Figure \ref{fig:BiasSignSims} plots the percentage of times that four phenomena of interest occur as a function of the average pairwise correlations between the pollutants. The first 3 panels demonstrate that phenomena 1--3 occurred more frequently when the average pairwise correlations between the pollutants were larger, and they even tended not to occur when the average pairwise correlations between the pollutants were negative. The fourth panel, shows that while the MEB tended to be nonnegative for nonnull error-prone pollutants, it was more likely to be negative when the average pairwise correlations between the pollutants were high. This result is consistent with phenomenon 4 occurring often, especially for large average pairwise pollutant correlations, and suggests that the attenuation bias is typically larger than the negative additive bias, leading to a net positive bias.

\begin{figure}[!t]
    \centering
    \includegraphics[width=0.95 \hsize]{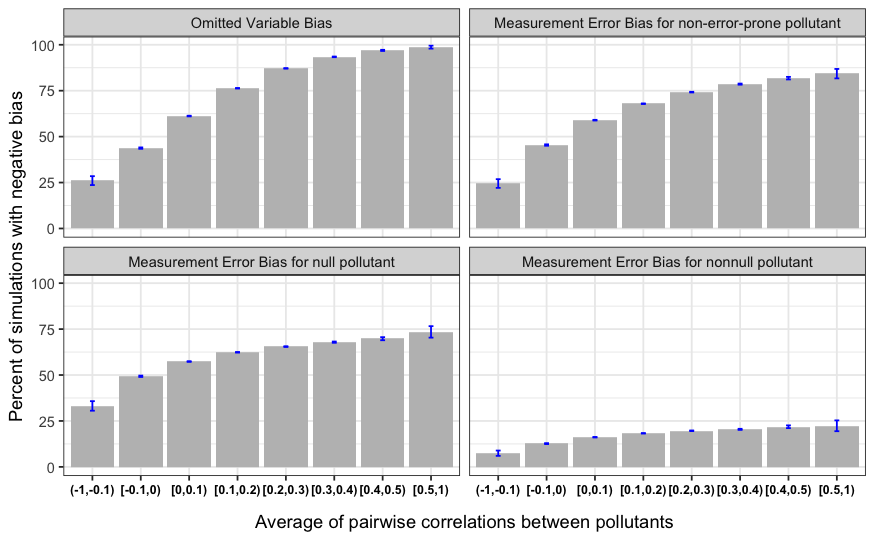}
    
    \caption{Proportion of Monte Carlo simulations with negative bias as a function of average pairwise correlation among the pollutants. We binned the $10^6$ Monte Carlo simulations into 8 categories, based on the average pairwise correlation between pollutants $\bar{\rho}$. The first bin is for simulations where $\bar{\rho} < -0.1$, the subsequent six bins are for simulations where $\bar{\rho}$ lies in a particular interval of length $0.1$ indicated by the x-axis labels, and the final bin is for simulations where $\bar{\rho}>0.5$. For each bin, the percentages of times $[\bbzOm]_5<0$ (top-left), $[\bbME]_5<0$ (top-right),  $[\bbME]_{10}<0$ (bottom-left) and $[\bbME]_{6}<0$ (bottom-right) are depicted. The error bars depict $\pm 1.96$ standard errors for the estimates of the proportion of times the bias is negative in each bin, where the uncertainty is due to the finite number of Monte Carlo trials run. The bins where $\bar{\rho} < -0.1$ and $\bar{\rho}>0.5$ had larger standard errors, because these events only occurred $1{,}263$ and $769$ times, respectively, out of the $10^6$ Monte Carlo simulations. }

    \label{fig:BiasSignSims}
\end{figure}

These results corroborate our claims that phenomena 1--4 tend to occur under the No Benefit and Pairwise $\PC$ Assumptions. The exact percentages presented in Table \ref{table:BiasSignSims} and Figure \ref{fig:BiasSignSims}  depend on the way we generated random covariance matrices of $(Z,X,W)$ and the relevant entries in $\beta$ in each simulation. The distribution of covariance matrices considered here, while having a broad support, will not be the same as the distribution of covariance matrices that investigators will encounter in practice. In the next section we study the OVB and MEB using a dataset and regression models that investigators would use in practice.

\section{Investigation of bias direction on a pollutant and crop yield dataset}\label{sec:ValidationOnDataset}

In this section we introduce a pollutant and crop yield dataset from the Midwestern United States and demonstrate that negative OVB and MEB tend to occur in our motivating application of assessing the association between air pollutants and crop yields using linear regression. In particular, we develop a validation scheme to demonstrate phenomena 1 and 2 that the perfectly measured pollutants tend to have negative OVB and MEB. We do not validate phenomena 3 and 4 about the direction of the MEB for null versus nonnull error-prone pollutants because it is not known which pollutants have no effects on crop yields.

\subsection{Dataset description}\label{sec:PanelDataset}

We assembled a county-level panel dataset in the Midwestern United States that spanned the years 2002--2019 that was similar to the one used in \cite{LobellBurneyERL2021} to study the impacts of air pollutants on corn and soybean yields. Our dataset comes from the same 9 states considered in \cite{LobellBurneyERL2021} and comprises a majority of the United States Corn Belt, a region where nearly 30\% of the global corn and soybean crop production occurs \citep{USDAWASDE}. Our study timespan was 2002--2019 because the error-prone predictions of air pollutants that we used were only available for those years.

Our outcome variables were annual county-level corn yields and soybean yields downloaded from the \cite{NASSQuickStats} Quick Stats webpage. The crop yield data is based on phone surveys of farmers in each county, and is only released in counties and years in which sufficiently many farmers responded to the survey. The yield data from 2010 is depicted in Figure \ref{fig:PollutantsAndCropYields} (top).

Daily air pollution data of $\CO$, $\NitrogenDi$, $\Ozone$, $\PMSmall$, $\PMBig$ and $\SulfurDi$ concentrations were downloaded from the Air Quality System (AQS) website of the \cite{EPA_AQS_Data}. The EPA's AQS is a database that contains air pollution data from a network of monitors. Figure \ref{fig:monitor_colocation} depicts the location of the monitors for our pollutants of interest in our study region that were active in 2010. We converted the daily resolution pollution data to yearly resolution data by taking the average of the daily mean pollutant concentrations across the 92 day June--August period. June--August corresponds to the peak of the growing season of corn and soybean crops in this region and is the same time period across which \cite{LobellBurneyERL2021} aggregated daily pollutant concentration data. %\cite{LobellBurneyERL2021} also aggregated AQS monitoring data across the same June--August time period; however, in contrast to the June--August average of daily maximum concentrations that they consider, we instead consider the June--August average concentration because because the error-prone proxy predictions of air pollutants that we used had daily average concentration estimates but not daily maximum concentration estimates.
For $\Ozone$ only, instead of considering the average concentration during the June--August period, we computed the average of the daily maximum of average concentrations across each 8-hour window, as this is a metric with which $\Ozone$ is commonly regulated (e.g. \cite{WHO_guidelines2021}), measured, or estimated (e.g. \cite{GRC_INLAPaper}). Figure \ref{fig:PairwiseCorr} depicts the sample pairwise correlations between the aggregated June--August pollution measurements for different pollutants, with each sample consisting of a unique monitoring site and year pair. Finally, we aggregated the monitoring pollutant data to the county level by taking averages across all monitoring sites within a county, rendering the pollution data at the same spatial and temporal resolution as the outcome variables.

Because many counties in our study region do not have AQS data for some or all of our pollutants of interest, we also used surface predictions of average pollutant concentrations from version 5.32 of the EPA's Community Multiscale Air Quality (CMAQ) model which were downloaded from the EPA's website \citep{EPA_CMAQ_Data}. CMAQ is a 3 dimensional chemical transport model that uses pollution emissions inventories and meteorological data as inputs to predict surface level air pollution concentrations \citep{CMAQ_ReviewPaper}. CMAQ has been used to study the impact of air pollution on crop yields in regions were pollution monitoring data is sparse \citep{Yi2016_OzoneChinaWheatObservational}. The prediction errors in CMAQ and their discrepancies with AQS monitor data are well documented (e.g. \cite{AQS_CMAQDataFusion}), and therefore, we treat the CMAQ data as an error-prone proxy for pollutant concentrations. The CMAQ-based pollutant concentration estimates that we downloaded were at 12 km $\times$ 12 km spatial resolution and a daily temporal resolution. In order to match the spatial and temporal resolution of our outcome variable we therefore averaged the pollutant concentrations across each county and across each June--August period. Figure \ref{fig:PollutantsAndCropYields} (bottom) shows the county-level averages of $\PMSmall$ during June--August of 2010 from the AQS monitoring data and from the CMAQ-based proxy data. Figure \ref{fig:AQSversusCMAQ} shows scatter plots of the AQS monitoring data versus the CMAQ-based proxies for all 6 pollutants with each point corresponding to a county-year pair in our dataset. %\todo{Mention \cite{LobellBurneyERL2021} don't use CMAQ?}

\begin{figure}[!t]
    \centering
    \includegraphics[width=0.95 \hsize]{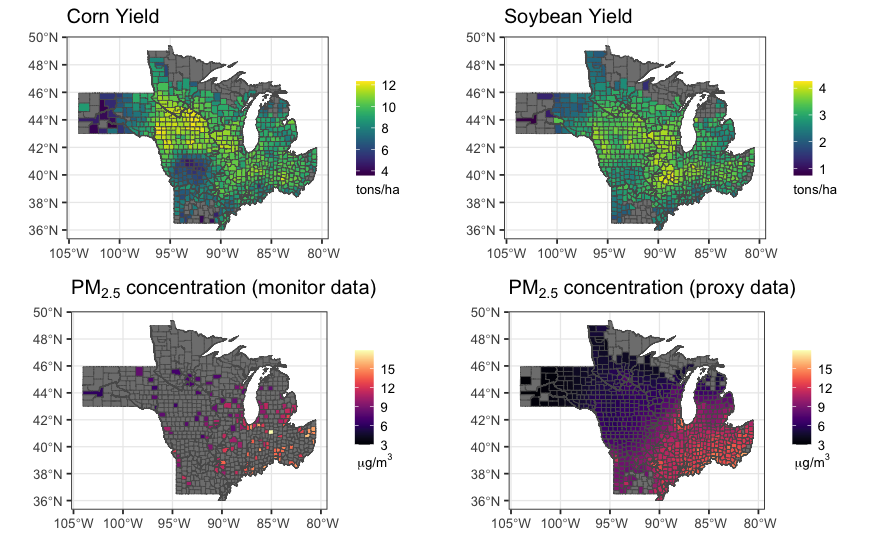}
    \caption{Visualization of panel dataset in 2010. This figure depicts the county-level averages of crop yields in 2010 for corn (top left) and soybean (top right). It also depicts the average June--August concentration in 2010 of $\PMSmall$ from AQS monitoring data (bottom left) and from the CMAQ-based proxy (bottom right). Counties colored in dark grey indicate missing data for either crop yields or monitoring data. Observe that many counties are missing monitoring data and some are missing crop yield data. The proxy data is available in every county but only counties with crop yield data are depicted. While not shown in this figure, our dataset also includes sparsely measured monitoring data and widely available proxy data for 5 other pollutants, 4 weather covariates, and data for each year between 2002 and 2019, inclusive.}
    \label{fig:PollutantsAndCropYields}
\end{figure}

\begin{figure}[!t]
    \centering
    \includegraphics[width=0.95 \hsize]{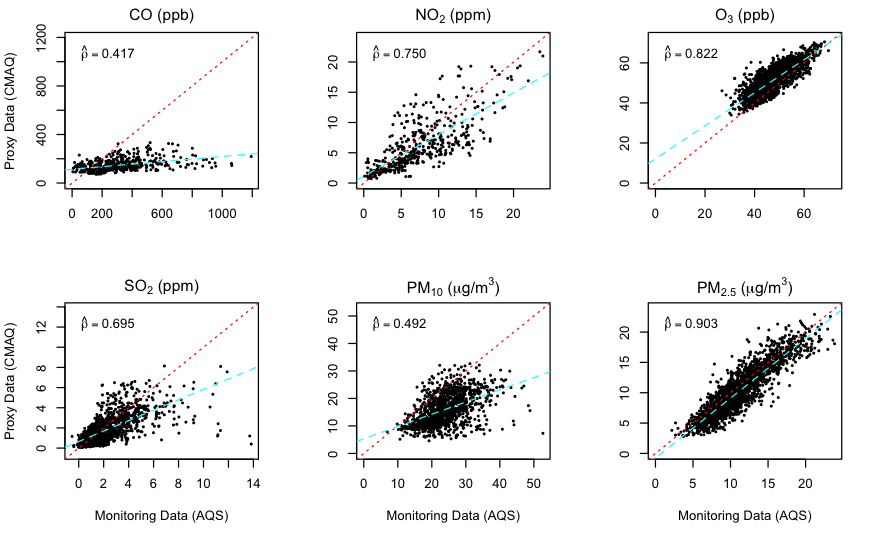}
    \caption{ Scatter plots of monitoring data versus proxy data for each pollutant. Each point in the scatter plot corresponds to a unique county-year pair in our dataset and depicts the average June--August concentration in that county and year as measured by the AQS pollution monitors and as estimated by the CMAQ-based proxy data. The scatter plots show a line through the origin with slope 1 (red, dotted), the best fit line (blue, dashed), and the sample Pearson correlations (top left text).}
    \label{fig:AQSversusCMAQ}
\end{figure}

We used the same four weather covariates as seen in \cite{LobellBurneyERL2021}: total April--August precipitation, total April--August precipitation squared, growing degree days, and extreme degree days. Growing degree days is a measure of the accumulated temperature exceeding 8 °C but below 30 °C  across the growing season (April--September) while extreme degree days is a measure of the accumulated temperature above 30 °C. A more explicit formula for these metrics and a description of how they are calculated using daily minimum and maximum temperatures can be found in Appendix A of \cite{RotationPaper}. Growing degree days and extreme degree days with a 30 °C threshold are commonly used covariates for modelling corn and soybean yields because there is evidence that temperature increases above 30 °C damage corn and soybean yields whereas temperature increases below 30 °C are slightly beneficial to corn and soybean yields \citep{SchlenkerAndRoberts}. We calculated annual values for these 4 weather covariates using monthly precipitation data and daily temperature extrema data from PRISM \citep{PRISM_original} which were accessed via Google Earth Engine \citep{PRISM_DailyData_VIA_GEE,PRISM_MonthlyData_VIA_GEE}, and computed county-level averages of these weather covariates by averaging across a random sample of 30 points in each county.

Let $i$ denote the county, $t$ denote the year, $Z_{it}$ denote the vector of 4 weather covariates, $X_{it}$ denote the vector of the 6 pollutant concentrations measured from monitors, and $Y_{it}$ denote the crop yields from the datasets described above. We suppose the investigator is interested in fitting the following regression with county level fixed effects $c_i$
\begin{equation}\label{eq:GeneralPanelRegressionModel}
    \log(Y_{it}) = \beta^\tran X_{it} + \gamma^\tran Z_{it} +\lambda t +c_i +\varepsilon_{it}.
\end{equation} To motivate this linear model, note that county level (or state level) fixed effects are common in the pollutant and crop yield literature \citep{BurneyIndia2014OzoneBlackCarbon,ColumbiaOzone2019,LobellBurneyERL2021,Da2022} and so is taking the log of the yield \citep{BurneyIndia2014OzoneBlackCarbon,ColumbiaOzone2019,HeEtAlOzoneMLpredsVsCropYieldsChina}. We include a linear time trend to avoid a spurious finding because in the Midwestern US, crop yields have been increasing over time while pollutant concentrations have been decreasing over time. We suppose that the investigator is most interested in estimating $\beta$ and consider the direction of the bias when they either omit components of the vector $X_{it}$ that are sparsely measured or replace a subset of the components of $X_{it}$ with their CMAQ-based proxies.

\subsection{Validation of negative biases}\label{sec:Validation Scheme}

Ideally, we would fit the linear regression specified by Equation \eqref{eq:GeneralPanelRegressionModel} and check the bias direction when omitting some of the 6 pollutants or when replacing them with their CMAQ-based proxies. However, this is not feasible because we do not know the true regression coefficients beforehand and the dataset cannot give us precise estimates of them, since only 193 county-year pairs have data for all 6 pollutants. Therefore, we use a validation scheme where we estimate the population regression coefficients in regressions with only 2 pollutants out of the 6. This way, we can use larger samples to get more precise estimates of the coefficients for validation. In each regression with 2 pollutants, we empirically estimate the OVB due to dropping one of the two pollutants or the MEB due to replacing one of the two pollutants by its CMAQ-based proxy. %To further increase the power of our analysis we aggregate the OVB and MEB estimates across all possible $6 \times 5 \times 2$ ordered pairs of pollutants and the outcome crop.

In particular for each $j,j' \in \{ \CO,\NitrogenDi,\Ozone,\PMBig,\PMSmall,\SulfurDi \}$ such that $j \neq j'$ and for each $k \in \{ \text{corn}, \text{soybean} \}$, we let $j$ denote the main pollutant, $j'$ the control pollutant, and estimate the OVB and MEB for pollutant $j$ as follows. We let $X_{it}^{(j)}$ and $X_{it}^{(j')}$ denote the AQS-based measurements of pollutant concentrations of $j$ and $j'$, respectively, $W_{it}^{(j')}$ denote the CMAQ-based estimate of the concentration of pollutant $j'$, and $Y_{it}^{(k)}$ denote the average yield for crop $k$ based on USDA survey data in county $i$ and year $t$. We then consider only the susbset of $(i,t)$ pairs for which $X_{it}^{(j)}$, $X_{it}^{(j')}$, and $Y_{it}^{(k)}$ measurements are available and run the following three regressions on this subset of the data. 
\begin{enumerate}
    \item \textit{"Ground truth" regression}:   $\log(Y_{it}^{(k)}) = \beta_{\text{GT}}  X_{it}^{(j)} +\alpha_1 X_{it}^{(j')}  + \gamma_1^\tran Z_{it} +\lambda_1 t + c_i +\varepsilon_{it}$,
    \item \textit{Omitted pollutant regression}: $\log(Y_{it}^{(k)}) = \beta_{\text{OM}}  X_{it}^{(j)}  + \gamma_2^\tran Z_{it} +\lambda_2 t + c_i +\varepsilon_{it}$, and
    \item \textit{Error-prone proxy regression}: $\log(Y_{it}^{(k)}) = \beta_{\text{ME}}  X_{it}^{(j)} +\alpha_3 W_{it}^{(j')}  + \gamma_3^\tran Z_{it} +\lambda_3 t + c_i +\varepsilon_{it}$.
\end{enumerate} 
Letting $\hat{\beta}_{\text{GT}}$, $\hat{\beta}_{\text{OM}}$, and $\hat{\beta}_{\text{ME}}$ be the point estimates for the regression coefficient of pollutant $j$ from the above 3 regressions, we define the point estimate of the OVB to be $\hat{\beta}_{\text{OM}}-\hat{\beta}_{\text{GT}}$ and the point estimate of the MEB to be $\hat{\beta}_{\text{ME}}-\hat{\beta}_{\text{GT}}$. We also record the t-statistics corresponding to $\hat{\beta}_{\text{GT}}$, $\hat{\beta}_{\text{OM}}$, and $\hat{\beta}_{\text{ME}}$.

Our theoretical results predict that the point estimates for the OVB and MEB should be negative; however, in each of the $60=6 \times 5 \times 2$ combinations of main pollutant, control pollutant, and outcome crop, the $95 \%$ confidence intervals of $\hat{\beta}_{\text{GT}}$ always overlapped with those of $\hat{\beta}_{\text{ME}}$ and $\hat{\beta}_{\text{OM}}$. Therefore, in order to demonstrate that there was a tendency towards negative OVB and MEB in our dataset, we aggregated the OVB and MEB estimates across the $60$ combinations considered using 3 different summary statistics. First, we counted the number of combinations out of 60 that had negative OVB and MEB point estimates. Second, we calculated the average difference in t-statistics across the 60 combinations (to study the OVB direction we considered the difference in t-statistics corresponding to $\hat{\beta}_{\text{OM}}$ minus that corresponding to $\hat{\beta}_{\text{GT}}$ whereas to study the MEB direction we considered the difference in t-statistics corresponding to $\hat{\beta}_{\text{ME}}$ minus that corresponding to $\hat{\beta}_{\text{GT}}$). We chose this approach because \cite{ackermann2009general} found that averaging $t$-statistics was a  very powerful way to aggregate over genes in gene sets. Third, we calculated the average OVB and average MEB point estimates across the 60 combinations after dividing the concentration for each pollutant by the WHO 2021 daily target (see Table 3.26 of \cite{WHO_guidelines2021}) for that pollutant, so that each pollutant would be on a comparable scale and be in units of WHO 2021 targets (we note that the outcome variable of log yield is unitless and we therefore do not need to rescale the yields). See Appendix \ref{sec:WHO_rescaling} for details about how we rescaled the pollutant concentrations. These 3 summary statistics for both the OVB and MEB are depicted as purple dots in the first and third rows of of Figure \ref{fig:BootstrapSummariesFigure}, respectively. For both bias types there were more negative biases than positive, a negative mean $t$-statistic difference and a negative mean coefficient difference over the 60 combinations considered.
%These results are consistent with our theoretical findings and show that across the 60 combinations considered, the OVB and MEB for the perfectly measured pollutant tend to be negative. 

\begin{figure}[!t]
    \centering
    \includegraphics[width=0.95 \hsize]{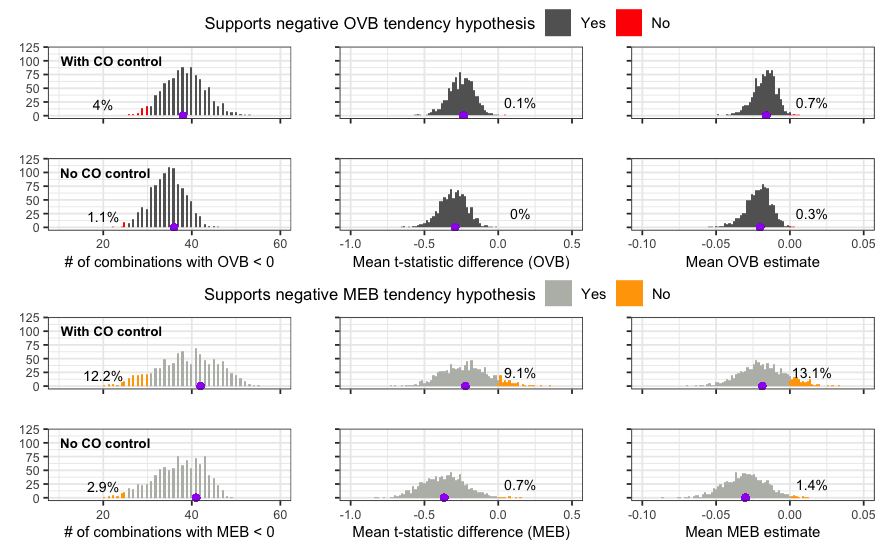}
    \caption{Histogram of summary statistics for the $1{,}000$ bootstrap trials. The 1st and 3rd row of histograms are summary statistics computed accross all 60 primary pollutant, control pollutant and outcome crop combinations whereas the 2nd and 4th row of histograms are summary statistics computed across the 50 combinations in which $\CO$ was not a control pollutant. In each histogram the purple dot denotes the value of the summary statistic computed on the actual data sample. Histogram bars corresponding to summary statistic values that support (do not support) our hypotheses of negative OVB and MEB tendencies are colored in dark grey (red) and light grey (orange), respectively. The percent of bootstrap simulations where the summary statistic value does not support the negative OVB or MEB tendency hypothesis is printed on each panel. In particular for the right two columns this percentage is the percentage of bootstrap simulations with a nonnegative summary statistic and in the first column this is the percentage of bootstrap simulations in which at most 30 out of 60 (in rows 1 and 3) or at most 25 out of 50 (in rows 2 and 4) combinations had a negative OVB or MEB. }
    \label{fig:BootstrapSummariesFigure}
\end{figure}

In order to demonstrate that these tendencies are statistically significant, we run $1{,}000$ bootstrap trials. In each bootstrap trial, we first resample with replacement the $767$ counties that have at least one year of corn or soybean data, and then on the bootstrapped dataset we conduct the same analysis described above of estimating the OVB and MEB for the 60 main pollutant, control pollutant and outcome crop combinations. We compute the same three summary statistics for both the OVB and MEB across the 60 combinations, and the distribution of these summary statistics across the $1{,}000$ bootstrap trials can be seen in the first and third rows of Figure \ref{fig:BootstrapSummariesFigure}. For each bootstrap trial, we also compute the three summary statistics for both the OVB and MEB, across the 50 main pollutant, control pollutant and outcome crop combinations in which $\CO$ is not the control pollutant and display the results in the second and fourth rows of Figure \ref{fig:BootstrapSummariesFigure}.

Our results show a consistent, statistically significant tendency towards negative OVB. They also suggest a tendency towards negative MEB, with statistical significance in the cases where only the 50 combinations without $\CO$ as a control pollutant are considered. Because the outcome variable is converted to the log scale and the pollutants are rescaled to be in terms of WHO targets, the x-axis values in the third column of Figure \ref{fig:BootstrapSummariesFigure} can roughly be interpreted as the average OVB or MEB in estimates of the effect of increasing average growing season concentrations of a pollutant by one daily WHO target on the $\%$ change in crop yield. The results therefore suggest that on average studies can overestimate the damage of increasing air pollution by 1 WHO target on crop yields by roughly 1--5\%.

Our empirical results are even more consistent with our theoretical results when we ignore combinations in which $\CO$ is a control pollutant, and there are a couple reasons why cases that involve $\CO$ as a control pollutant do not support our theoretical findings. First, $\CO$ may be beneficial to crop yields at the concentrations seen in our dataset and therefore violate the No Benefit assumption that was critical to our theoretical results. \cite{PlantGrowthCOReviewPaper} reviews a literature of experiments which demonstrate that small amounts of exposure to either ambient or aqueous $\CO$ benefits various physiological processes in plants. While $\CO$ is harmful to both plants and animals at high concentrations, in our dataset the $\CO$ concentration is far below the WHO 2021 targets, and therefore could be beneficial at the concentrations seen in our dataset. For other pollutants in our dataset, the concentrations do not consistently fall far below the WHO targets (Figure \ref{fig:AQSvsWHOtargets}) and render it less plausible that the other pollutants violate the No Benefit assumption and help crop yields at the concentrations seen in the data. Second, as is seen in Figure \ref{fig:AQSversusCMAQ} as well as in \cite{AQS_CMAQDataFusion}, the CMAQ-based proxy for $\CO$ has particularly poor agreement with the monitoring data. Therefore, the measurement error for the $\CO$ proxy may be so irregular that our theoretical results suggesting a tendency of negative MEB for perfectly measured pollutants do not hold in the case where $\CO$ is the control pollutant.

\section{Discussion} \label{sec:Discussion}

We discuss several limitations of the assumptions used for obtaining theoretical results in our study. First, all formulas and non data-based analyses regarding MEB in this paper assume that the regression residuals $\varepsilon_i$ are uncorrelated with the measurement errors $W_i-X_i$ (see Condition (ii) in Section \ref{sec:Setup}). This assumption may not hold for remote sensing-based or machine learning-based predictions $W$ of air pollutant concentrations $X$. %, and handling such cases is out of the scope of the present article. 
Second, our analyses regarding the direction of the OVB and the MEB rely heavily upon the assumption that the regression coefficients corresponding to the pollutants are all nonpositive. %This assumption, while seemingly realistic, may not hold for all of the motivating applications of interest. 
There are exceptions, where some pollutants can have a positive association with crop yields. For example, nitrate deposition associated with nitrate aerosols can have a direct positive effect on crop yields \citep{NitrogenDepositionChina}, and $\CO$ can benefit plant growth at low concentrations \citep{PlantGrowthCOReviewPaper}. Third, our analyses rely upon the assumption that all pollutants have positive pairwise correlations. While we expect this assumption to be easily testable and true in most motivating settings, analyses using panel datasets often leverage two-way fixed effects models to control for space and time. Regressing out two-way fixed effects can weaken or even remove positive correlations between pollutants, as is seen in Table S1 of \cite{LobellBurneyERL2021}.

Given these limitations, we cannot be sure that any given  study in the literature had negative OVB.  Perhaps some omitted pollutants had a positive coefficient for the outcome or had a negative partial correlation with a non-omitted pollutant conditional on the model's space, time, and weather controls. Nonetheless, we expect negative OVB to be the norm even in the presence of some weak negative dependencies between pollutants or some weakly positive regression coefficients, and indeed we found clear evidence of negative OVB in our motivating application on a real dataset in Section \ref{sec:ValidationOnDataset}.

Our findings are less conclusive regarding the direction of the MEB. Both our theoretical and simulation-based results rely on a correlational version of the nondifferential measurement error assumption and our theoretical results rely upon an even stronger assumption of classical measurement error that is uncorrelated across pollutants. Under these assumptions, we found theoretical guarantees that the MEB is negative for pollutants that are perfectly measured as well as for null pollutants and that these biases tended to be negative in the more general simulation settings. The MEB of nonnull pollutants tended to be positive in our simulations. Because an investigator is unlikely to know a priori whether a pollutant is null or nonnull and the extent to which the nondifferential measurement error assumption is violated, we cannot be certain of the direction of MEB in specific other papers. For null pollutants and perfectly measured pollutants, we do, however, expect negative MEBs to be more common than positive ones. Indeed, we found evidence of negative MEB for perfectly measured pollutants in our motivating application on a real dataset in Section \ref{sec:ValidationOnDataset} without making any assumptions about the structure of the measurement error.

Our results on the direction of the OVB and MEB have the following two consequences.
First, for perfectly measured pollutants, Proposition \ref{prop:OmBiasNegConds}, Corollary \ref{cor:NegBiasNoME}, our simulations, and our data-based validation all suggest that the OVB and MEB will tend to be negative. Hence if there is one perfectly measured pollutant such as $\Ozone$ and the remaining relevant pollutant concentrations only have error-prone estimates $W$, an investigator cannot eliminate bias induced by omitted variables by simply adding $W$ as a covariate in their regression model. Second, there is a common myth that measurement error simply biases regression coefficients of error-prone covariates towards zero. We show, however, that under certain assumptions, some of which are common in the measurement error literature and some of which are realistic to our motivating application, the coefficient estimates corresponding to null pollutants are guaranteed to be negatively biased away from zero. 

We briefly discuss frequentist and Bayesian approaches to correct for MEB as well as datasets that could be useful to these ends. In settings where all variables of interest $(Y,Z,W,X)$ are available on a small sample while a much larger sample only has $(Y,Z,W)$ measurements, numerous methods have been proposed to correct for MEB \citep{CarrolStefanski06,FongAndTyler2021,Proctor2023RS_MI,Stephen2023PredictionPoweredInf}; however, these methods still have challenges in their application to our motivating setting. First, only the methods in \cite{CarrolStefanski06}, such as regression calibration, can be readily applied to settings where different components of $X$ are missing in different subsets of the small sample (as Figure \ref{fig:monitor_colocation} suggests would be the case in air pollutant studies). Second, with the exception of the method in \cite{Stephen2023PredictionPoweredInf}, these methods fail to remove the bias asymptotically in settings where Condition (ii) or the nondifferential measurement error assumption fail to hold. Third, in our motivating setting, because air pollutant monitors tend to be placed near cities rather than croplands (Figure \ref{fig:monitor_colocation}), and because pollution concentrations tend to be higher near cities than they are near farmlands, there is a shift in the distribution of $X$ between the locations where the measurements of $X$ are available and the target sample of interest. While \cite{Stephen2023PredictionPoweredInf} presents an approach for extending their method to correct for distribution shift in a similar setting, that approach would only apply to our setting if entries of $X$ were discrete variables. 

The aforementioned approaches are not in a formal causal inference framework. \cite{XiaoWuFrancesaDominiciMECorrection} develops a regression calibration-based approach that corrects for MEB in a formal causal inference framework to study the impact of air pollution on mortality; however, their method also relies upon the nondifferential measurement error assumption and focuses on a setting where there are $d=1$ error-prone pollutants.

Another appealing class of methods to correct for MEB are 2-stage Bayesian methods, in which one research group fits a Bayesian model to estimate the error-prone covariates (such as the Bayesian model for pollutant concentrations in the United States in \cite{GRC_INLAPaper}) and produces posterior samples for the error-prone covariates. A second research group then uses posterior samples from the first group's model to form a prior for a downstream Bayesian inference task, such as regression. \cite{Bayesian2StageMethodology} develop a 2-stage Bayesian method and also review previously existing 2-stage Bayesian methods, many of which can be naturally extended to settings where there are multiple error-prone pollutants. They compare the methods both in simulations and in an example using $\PMSmall$ concentration data, proxy estimates for $\PMSmall$ data, and a health outcome variable. 2-stage Bayesian methods are appealing because they provide a division of labor between groups that have expertise in modelling pollutant concentrations and groups that have expertise in modelling outcome variables.  However, a challenge with the 2-stage Bayesian methods is that the Bayesian estimates produced by the first research group may not be entirely geared towards the downstream task of regressing an outcome variable on pollutant concentrations, even if the predictions and uncertainty quantification pass the first group's validation tests in terms of $R^2$ and coverage. For example, the errors in the pollutant concentration estimates may be correlated with the outcome variable or other covariates that end up being used in the downstream task, leading to bias in the downstream regression. Another subtle issue is that the predicted pollutant concentrations may be much more correlated with each other than are the actual pollutant concentrations, which could lead to inaccurate inferences in the downstream regression task. 

In light of these challenges, air pollutant monitoring sites that accurately measure a full suite of pollutants and that are placed near farmlands would greatly benefit research into the association between air pollutants and crop yields. 
However, it would take a number of years for these benefits to be conferred to research efforts. In the shorter term, methods are needed that can correct for MEB in settings where there is differential measurement error and distributional shift between monitoring sites (which tend to be near cities) and the target sample (e.g., the sample of farmland), and where different monitoring sites are missing different subsets of the pollutants of interest.

\section*{Acknowledgments}
This work was supported by a Stanford Interdisciplinary Graduate Fellowship and the National Science Foundation under grant DMS-2152780. The authors wish to thank Wenlong Gong and Brian Reich for helpful discussions.

\bibliographystyle{apalike}
\bibliography{PollutantsAndCropYields}

\appendix
\setcounter{table}{0}
\renewcommand{\thetable}{A\arabic{table}}
\setcounter{figure}{0}
\renewcommand{\thefigure}{A\arabic{figure}}
\section{Notation for Appendix}\label{sec:notationAppendix}

These appendices have our derivations for omitted variable bias (OVB) and measurement error bias (MEB) formulas presented in the main text, comparisons between the OVB and MEB, proofs of Proposition \ref{prop:OmBiasNegConds} and Theorem \ref{thoerem:PollutantBiasOmegaSigns}, and further details on our simulations.
In this section, we present some notation used throughout the appendices. Throughout the appendices, we use $\bm{1}_n$ to denote a vector of $n$ ones, $\bm{1}$ to denote a vector of ones when the dimension is understood from context, and $I_n$ to denote an $n \times n$ identity matrix.

We have observations indexed by $i=1,\dots,N$. Recall that we use random variables $Z_i\in\real^p$, $X_i\in\real^d$, $W_i\in\real^d$ and $Y_i\in\real$ representing covariates measured without error, the true value of pollutants, the value of proxy measures for those pollutants, and the response, respectively. Sample averages of these quantities are denoted by $\bar Z$, $\bar X$, $\bar Y$ and $\bar W$. Components are denoted by superscripts, as in $Z^{(j)}$ for the $j$'th covariate. We use $o_p(1)$ to denote any term that converges to $0$ in probability as the number of samples $N \to \infty$.

When deriving formulas for the OVB and MEB in Appendix \ref{sec:formula_deriv}, it is convenient to use centered data matrices of these variables. We define $\cz\in\real^{N\times p}$, $\cx\in\real^{N\times d}$ and $\cw\in\real^{N\times d}$ to be matrices with $i$'th rows $Z_i-\bar Z$, $X_i-\bar X$, and $W_i -\bar W$ respectively. Observe that since $\cz$, $\cx$, and $\cw$, are centered, $\cz^\tran \bm{1}=0$, $\cx^\tran \bm{1}=0$, and $\cw^\tran \bm{1}=0$. We also let $\cy\in\real^{N} = (Y_1,\ldots,Y_N)$ denote the vector of responses %for the $N$ 
and let $\varepsilon = (\varepsilon_i)_{i=1}^N$ be the vector of the error terms in the regression model. Recall that for $i=1,\ldots,N$, $\varepsilon_i \equiv Y_i -\alpha -Z_i^\tran \beta_Z -X_i^\tran \beta_X$, and note that by definition of $\varepsilon_i$ and by definition of $(\alpha,\beta_Z,\beta_X)$ in~\eqref{eq:def_reg_coef}, $\varepsilon_i$ is uncorrelated with $Z_i$ and $X_i$. Recalling that under Condition (ii), $\varepsilon_i$ is uncorrelated with $W_i-X_i$, it follows that under Condition (ii),  $\varepsilon_i$ must be uncorrelated with $W_i$ as well.

%which is assumed to satisfy Condition (ii) in the main text. 

In Appendix \ref{sec:OmBiasVsMEbias}, we investigate the relationship between OVB and MEB under various assumptions on $p$, $d$ and the measurement error structure. For any random variable $U$ and $V$ in that Appendix, it is convenient to let $\rho_{UV}$ denote the correlation between $U$ and $V$, and let $\sigma_U$ and $\sigma_V$ denote the standard deviations of $U$ and $V$ respectively. So for example in the case where $d=1$, $X$ is a random variable and $\rho_{X Z^{(k)}} =\corr(X,Z^{(k)})$ and $\sigma_X^2=\var(X)$. 

\section{Formulas for OVB and MEB}\label{sec:formula_deriv}

\subsection{OVB formula derivation}\label{sec:ovbderivation}

Here we suppose the investigator estimates $\beta_Z$ by regressing $Y$ on $Z$ with OLS; that is they estimate $\beta_Z$ with $\hat{\beta}_Z = (\cz ^\tran \cz)^{-1} \cz^ \tran \cy$. Thus under Conditions (i), %letting $o_p(1)$ denote a term that converges to $0$ in probability as the number of samples $N \to \infty$, it follows that 
$$\begin{aligned}
\hat{\beta}_Z & = (\cz ^\tran \cz)^{-1} \cz^ \tran \cy
\\ & = (\cz ^\tran \cz)^{-1} \cz^ \tran \big( \alpha \bm{1} + \cz \beta_Z + \cx \beta_X +\varepsilon \big) 
\\ & = 0+ \beta_Z + \big( \frac{1}{N} \cz ^\tran \cz \big)^{-1} \Big( \big(\frac{1}{N} \cz^ \tran \cx \big) \beta_X + \big( \frac{1}{N} \cz^ \tran \varepsilon \big) \Big)
\\ & \stackrel{\text{(i)}}{=} \beta_Z + \big( \cov(Z,Z) +o_p(1) \big)^{-1} \Big( \big( \cov(Z,X) +o_p(1) \big) \beta_X + \cov(Z_i,\varepsilon_i) +o_p(1)  \Big)
\\ & = \beta_Z + \big( \cov(Z,Z) +o_p(1) \big)^{-1} \Big( \big( \cov(Z,X) +o_p(1) \big) \beta_X + 0 +o_p(1)  \Big)
\\ & = \beta_Z + \big[\cov(Z,Z) \big]^{-1} \big[ \cov(Z,X)\big] \beta_X +o_p(1)
\end{aligned}$$
%\\ & = \beta_Z + A^{-1} B \beta_X +o_p(1).
%\end{aligned}$$

Above, in the fourth step, we used Condition (i) and the fact that $\cx$ and $\cz$ are the data matrices of centered versions of $X$ and $Z$, and in the fifth step we use the fact that $\corr(Z_i,\varepsilon_i)=0$. Since the term $o_p(1)$ converges to zero in probability as $N \to \infty$, $$\text{plim}_{N \to \infty} ( \hat{\beta}_Z - \beta_Z )=\big[\cov(Z,Z) \big]^{-1} \big[ \cov(Z,X)\big] \beta_X = A^{-1} B \beta_X.$$

\subsection{MEB formula derivation}\label{sec:mebderivation}

Here we suppose the investigator estimates $\beta=(\beta_Z,\beta_X)$ by regressing $Y$ on $(Z,W)$ with OLS; that is letting $\cmw = \begin{bmatrix} \cz & \cw \end{bmatrix}$ be the centered data matrix of $(Z,W)$ they estimate $\beta$ with $\hat{\beta} = (\cmw ^\tran \cmw)^{-1} \cmw^ \tran \cy$.  
It is convenient to let $\cmx = \begin{bmatrix} \cz & \cx \end{bmatrix}$ be the centered data matrix of $(Z,X)$ and $\ce= \cmx-\cmw$ be the centered covariate error matrix. 
%Now note that by 
By Condition (i) and the fact that $\cz$, $\cx$, and $\cw$ are centered matrices, 
%and that $\Sigma_M$ is the covariance matrix of available covariates $(Z,W)$, while $\Sigma_{M,E}$ is the covariance of $(Z,W)$ with its error as an estimate of $(Z,X)$, 
$\frac{1}{N} \cmw^ \tran \cmw \xrightarrow{p} \Sigma_M$ and $\frac{1}{N} \cmw^ \tran \ce \xrightarrow{p} \Sigma_{M,E}$, where $\Sigma_M$ and $\Sigma_{M,E}$ are defined at~\eqref{eq:thosecovs}. Thus under Conditions (i) and (ii), %letting $o_p(1)$ denote a term that converges to $0$ in probability as the number of samples $N \to \infty$, it follows that 

$$\begin{aligned}
\hat{\beta} & = (\cmw^ \tran \cmw)^{-1} \cmw^\tran \cy
\\ & = (\cmw^ \tran \cmw)^{-1} \cmw^\tran (\alpha \bm{1} + \cz \beta_Z +\cx \beta_X+\varepsilon)
\\ & = 0+(\cmw^ \tran \cmw)^{-1} \cmw^\tran ( \cmx \beta +\varepsilon)
\\ & = (\cmw^ \tran \cmw)^{-1} \cmw^\tran ( \cmw \beta + \ce \beta +\varepsilon)
\\ & = \beta + \Big( \frac{1}{N} \cmw^ \tran \cmw \Big)^{-1}  \Big( \frac{1}{N} \cmw^\tran \ce \beta + \frac{1}{N} \cmw^\tran \varepsilon \Big)
\\ & \stackrel{\text{(i)}}{=} \beta + \big( \Sigma_M +o_p(1) \big)^{-1}  \bigg( \big(\Sigma_{M,E} +o_p(1) \big) \beta + \cov \Big( \begin{bmatrix} Z_i \\ W_i \end{bmatrix}, \varepsilon_i \Big) +o_p(1) \bigg)
\\ & \stackrel{\text{(ii)}}{=} \beta + \big( \Sigma_M +o_p(1) \big)^{-1}  \Big( \big(\Sigma_{M,E} +o_p(1) \big) \beta + 0+o_p(1) \Big)
\\ & = \beta +  \Sigma_M^{-1} \Sigma_{M,E} \beta +o_p(1).
\end{aligned}$$

Above, in the sixth step, we used Condition (i) and the fact that $\cmw$ and $\ce$ are matrices with centered columns, and in the seventh step we use Condition (ii) coupled with the fact that $\varepsilon_i$ is uncorrelated with $Z_i$ and $X_i$. Since the term $o_p(1)$ converges to zero in probability it follows that %as $N \to \infty$ it follows that 
$$\bbME \equiv \text{plim}_{N \to \infty} ( \hat{\beta} - \beta )=\Sigma_M^{-1} \Sigma_{M,E} \beta,$$ confirming Equation \eqref{eq:MEbiasAll}.

We now use the above formula to derive the MEB vector $\bbzME$ for $\beta_Z$. Recalling that $$\Sigma_M  =\begin{bmatrix} A & C \\ C^\tran & G \end{bmatrix} \quad \text{and} \quad  \Sigma_{M,E} =\begin{bmatrix} 0 & B-C \\ 0 & F^ \tran -G \end{bmatrix},$$ %the above formula implies
$$\begin{aligned}
     \bbzME  & = \Big[ \Sigma_M^{-1} \Sigma_{M,E} \beta \Big]_{1:p}
     \\ & = \Bigg[ \begin{bmatrix} [\Sigma_M^{-1}]_{1:p,1:p} & [\Sigma_M^{-1}]_{1:p,(p+1):(p+d)} \\ [\Sigma_M^{-1}]_{(p+1):(p+d),1:p} & [\Sigma_M^{-1}]_{(p+1):(p+d),(p+1):(p+d)} \end{bmatrix} \begin{bmatrix} 0 & B-C \\ 0 & F^ \tran -G \end{bmatrix} \begin{bmatrix} \beta_Z \\ \beta_X \end{bmatrix} \Bigg]_{1:p}
     \\ & =  [\Sigma_M^{-1}]_{1:p,1:p} (B-C) \beta_X+ [\Sigma_M^{-1}]_{1:p,(p+1):(p+d)} (F^ \tran -G) \beta_X
   %   \\ & =  (A-C G^{-1} C^ \tran )^{-1} (B-C) \beta_X- (A-C G^{-1} C^ \tran )^{-1} C G^{-1} (F^ \tran -G) \beta_X
    \\ & =  (A-C G^{-1} C^ \tran )^{-1} (B  -CG^{-1} F^ \tran )\beta_X.
\end{aligned}$$
In the last step above, we apply the blockwise matrix inversion formula and cancel terms.

\section{Relationships between MEB and OVB}\label{sec:OmBiasVsMEbias}

Here we investigate whether the OVB is necessarily worse than the MEB. In particular, we check whether the entries of $\bbzOm$ must be at least as far away from zero as the corresponding entries of $\bbzME$ are. We investigate this by considering six different cases, which are summarized in Table  \ref{table:OMbiasVsMEbias}. In Appendices \ref{sec:OmBiasVsMEbias}.1--\ref{sec:OmBiasVsMEbias}.6, we show the results summarized in Table \ref{table:OMbiasVsMEbias}, with one case covered in each subsection.

In Cases 1--4 the entries of $\bbzOm$ are guaranteed to be at least as far away from zero as the corresponding entries of $\bbzME$ are. For Case 1, we present the well known result that under a Berkson measurement error model, $\bbzME=0$. For Case 2 where $d=1$, we prove that OVB is guaranteed to be worse than the MEB if there is no partial correlation between $Z$ and the error-prone proxy $W$, conditional on the true value $X$ of the error-prone variable. This partial correlation assumption in Case 2 can be thought of as a generalization of the classical measurement error assumption, because it is guaranteed to hold in classical measurement error settings. For Case 3 where $p=d=1$, we show that under some plausible
%reasonably loose 
assumptions on the relative sizes and signs of $\rho_{XW}$, $\rho_{ZX}$, and $\rho_{ZW}$, the OVB is guaranteed to be worse than the MEB. In Case 4, our results show that under a classical measurement error model, in the limit as the size of the measurement error goes to infinity, $\bbzME$ converges to $\bbzOm$.

\begin{table}[t]
\caption{Summary of cases considered to investigate whether OVB is always worse than the MEB. In all cases we assume Conditions (i)-(iii) so that formulas \eqref{eq:OmVarBias} and \eqref{eq:MEbiasZ} hold.  In Case 2, $\bm{\rho}_{WZ \mid X}$ denotes the vector of partial correlations between $W$ and $Z^{(j)}$ conditional on $X$, for $j=1,\dots,p$. In Case 3, SA stands for the ``sign assumption" that $\rho_{ZW}$ and $\rho_{ZX}$ have the same sign.}
\label{table:OMbiasVsMEbias}
% \centering works a bit differently from \begin{center} .. \end{center]
\centering
\begin{tabular}{  c  l  l  l  } 
\toprule
 Case  & Dimensions & Measurement error assumptions & 
 $\big|\mathrm{OVB}\big|>\big|\mathrm{MEB}\big|$\\
% \bbzOm$ worse than $\bbzME$?   \\ 
\midrule
%\#
 1 & $p \geq 1$, $d \geq 1$ & Berkson-type error  & Always: $\bbzME=0$  \\  [1ex] 
 2 & $p\geq 1$, $d=1$  & $\bm{\rho}_{WZ | X} =0$ & Always   \\ [1ex] 
 3 & $p=d=1$  & $\rho_{XW} > \vert \rho_{ZW} \vert, \vert \rho_{ZX} \vert >\vert \rho_{ZW} \vert$, SA  & Always  \\ [1ex] 
 4 & $p \geq 1$, $d \geq 1$ & Classical measurement error $\to \infty$  & Same: $\bbzME\,\to\,\bbzOm$  \\ [1ex] 
 5 & $d > p$ & Classical measurement error & Sometimes   \\ [1ex] 
 6 & $p=d=1$ & None & Sometimes  \\ [1ex] 
 \bottomrule
\end{tabular}

\end{table}

While our findings in Case 1--4 might suggest on their own that the OVB is always worse than the MEB, in Cases 5 and 6 we find that there are examples where the MEB is strictly worse than the OVB. In Case 5 where $d>p$, even if we are willing to assume classical measurement error, we provide a counterexample showing that the MEB can be strictly worse than the OVB. We believe that under classical measurement error, the MEB is only worse than the OVB for particular directions of the unknown vector $\beta_X$ such that the OVB is near $0$. In Case 6 where $p=d=1$ and no measurement error assumptions are imposed, we present a rather artificial counterexample where the OVB$=0$ and the MEB can be arbitrarily large.

\subsection{No bias under Berkson-type error}

In the Berkson-type error model we suppose that $X=W+E$, where $E$ is a random vector of errors with mean $0$ that is independent of $W$, $Z$, and $\varepsilon$. In this case note that $B=\cov(Z,W+E)=C $ and $F=\cov(W+E,W)=G$. Plugging this into the formula for $\Sigma_{M,E}$ yields that $\Sigma_{M,E}=0$, which by \eqref{eq:MEbiasAll} further implies that under the Berkson-type measurement error model $\bbME =0 \Rightarrow \bbzME=0$. This confirms the well known result that if the covariates have Berkson-type measurement errors it does not induce bias in the regression coefficient estimates.

\subsection{Case where $d=1$ and partial correlation of $W$ and $Z$ given $X$ is zero} 

Suppose $d=1$ and further we assume that for each $k \in \{1,\dots,p \}$ the partial correlation of $W$ and $Z^{(k)}$ given $X$, denoted by $\rho_{W Z^{(k)} \mid X}$, is equal to zero. We will show that when $\rho_{WZ^{(k)} \mid X}=0$ for all $k \in \{1,\dots,p \}$, then OVB is strictly worse than MEB. 

Since $d=1$, $\beta_{X}$, $\rho_{XW}$, $\sigma_X$ and $\sigma_W$ are scalars. Note that $F= \sigma_X \sigma_W \rho_{XW}, \ G=\sigma_W^2,\text{ and } D=\sigma_X^2.$ Further, since $d=1$, $C$ and $B$ are just vectors in $\mathbb{R}^{p}$. Thus for any $k \in \{1,\dots,p \}$ our assumption that $\rho_{WZ^{(k)} \mid X}=0$ implies that
$$0=\rho_{WZ^{(k)}}-\rho_{XW} \rho_{XZ^{(k)}}=\frac{1}{ \sigma_{Z^{(k)} }} \Big( \frac{C_k}{\sigma_W} - \rho_{XW} \cdot \frac{B_k}{\sigma_X} \Big)$$
and so $C_k = (\sigma_W \rho_{XW}/\sigma_X)B_k.$

Hence the assumption of no partial correlation between $W$ and each entry of $Z$, conditional on $X$ implies that $C = \sigma_X^{-1} \sigma_W \rho_{XW} B$. Plugging these results into Equation \eqref{eq:MEbiasZ}, $$\bbzME  =  \big( A- C G^{-1} C^ \tran \big)^{-1} (B - C G^{-1} F^ \tran ) \beta_X
 = \beta_X \Big(A -  \frac{CC^\tran }{\sigma_W^2}  \Big)^{-1} \big( B - \rho_{XW}^2 B \big).$$

By noting $C$ and $B$ are column vectors and applying the Sherman-Morrison formula,
$$\begin{aligned} \bbzME & = \beta_X  \Big( A^{-1} + \frac{A^{-1} C C^\tran  A^{-1}/\sigma_W^2 }{ 1  - C^\tran  A^{-1} C/\sigma_W^2 } \Big) (1-\rho_{XW}^2) B
\\ & = \beta_X (1-\rho_{XW}^2) \Big( A^{-1} + \frac{A^{-1} B B^\tran  A^{-1} ( \rho_{XW}^2/ \sigma_X^2) }{ 1  - B^\tran  A^{-1} B  ( \rho_{XW}^2/ \sigma_X^2) } \Big) B
\\ & = \beta_X (1-\rho_{XW}^2)  A^{-1} B  \Big( 1 + \frac{B^\tran  A^{-1} B}{ \sigma_X^2/\rho_{XW}^2  - B^\tran  A^{-1} B  } \Big)
%\\ & =  (1-\rho_{XW}^2)    \Big( 1 + \frac{B^\tran  A^{-1} B}{ \sigma_X^2/\rho_{XW}^2  - B^\tran  A^{-1} B  } \Big) \Big(  A^{-1} B \beta_X \Big)
\\ & =  \varrho(A,B,\rho_{XW},\sigma_X) \times \bbzOm,
\end{aligned}$$ where $\bbzOm=A^{-1}B \beta_X$ by Equation \eqref{eq:OmVarBias} and $$\varrho(A,B,\rho_{XW},\sigma_X)  \equiv (1-\rho_{XW}^2)    \Big( 1 + \frac{B^\tran  A^{-1} B}{ \sigma_X^2/\rho_{XW}^2  - B^\tran  A^{-1} B  } \Big).$$

For the case where $\rho_{XW}=0$, the above result implies that $ \bbzME  = \bbzOm$. We now consider the case where $\rho_{XW} \neq 0$. To show $ \bbzME$ is smaller than $\bbzOm$ in this case, it suffices to show that $\varrho(A,B,\rho_{XW},\sigma_X) \in [0,1)$. To see this, note that $\sigma_X^2 - B^\tran  A^{-1} B>0$ because it is the Schur complement of the covariance matrix of $(Z,X)$, which is positive definite by Regularity Condition (iii). By positive definiteness of $A$, this implies $$\sigma_X^2 > B^\tran  A^{-1} B>0.$$

In the case where $\rho_{XW} \neq 0$, $0 < \rho_{XW}^2 \leq 1$. Combining the above result and the definition of $\varrho$, it is clear that $\varrho(A,B,\rho_{XW},\sigma_X) \geq 0$. The above result also implies that

$$\rho_{XW}^2 > \frac{B^\tran  A^{-1} B}{\sigma_X^2/\rho_{XW}^2} =  \Big( \frac{B^\tran  A^{-1} B}{\sigma_X^2/\rho_{XW}^2  - B^\tran  A^{-1} B} \Big) \Big( 1+\frac{B^\tran  A^{-1} B}{\sigma_X^2/\rho_{XW}^2  - B^\tran  A^{-1} B} \Big)^{-1}.$$ Multiplying each side of the above inequality by the second factor on the RHS, reordering terms and adding one to each side implies that $$1 > 1 + \Big( \frac{B^\tran  A^{-1} B}{\sigma_X^2/\rho_{XW}^2  - B^\tran  A^{-1} B} \Big)  - \rho_{XW}^2 \Big( 1+\frac{B^\tran  A^{-1} B}{\sigma_X^2/\rho_{XW}^2  - B^\tran  A^{-1} B} \Big)=\varrho(A,B,\rho_{XW},\sigma_X).$$

Hence we have shown that in the case where $d=1$ and there is no partial correlation between $W$ and $Z$, conditional on $X$, OVB is worse than bias due to MEB, and in particular $$ \bbzME = \varrho(A,B,\rho_{XW},\sigma_X) \times \bbzOm  \text{ where } \varrho(A,B,\rho_{XW},\sigma_X) 
\in \begin{cases}
    [0,1) & \text{ if } \rho_{XW} \neq 0,
    \\ \{1 \}  & \text{ if } \rho_{XW} = 0. \end{cases}$$

\subsection{Case where $p=d=1$ and some measurement error assumptions}

Here we consider the case where $p=d=1$ under the % 
plausible 
assumption that %both 
$\rho_{XW} >  \vert \rho_{ZW} \vert >0$, while $\rho_{ZX}$ and $\rho_{ZW}$ have the same sign, with  $ \vert \rho_{ZX} \vert \geq \vert \rho_{ZW} \vert$.  This assumption is reasonable in many settings because $W$ is an error-prone estimate of $X$, so we would expect the correlation between $W$ and $X$ to be both positive and larger than the correlation between $W$ and $Z$. The assumption that $\rho_{ZX}$ and $\rho_{ZW}$ have the same sign is also reasonable because $W$ is an error-prone estimate of $X$. The assumption that $Z$ is more correlated with $X$ than with an error-prone estimate of $X$ will not always 
hold, %reasonable to assume, 
but for many measurement error structures, including for classical measurement error, it will hold. 

In this setting $A=\sigma^2_Z$, $B=\rho_{ZX} \sigma_Z \sigma_X$, $C= \rho_{ZW} \sigma_Z \sigma_W$, $F= \rho_{XW} \sigma_X \sigma_W$ and $G=\sigma^2_W$. Since $d=1$, $\beta_X$ is a scalar so by Equation \eqref{eq:MEbiasZ},
$$\bbzME  =  \frac{B-CF/G}{A-C^2/G} \beta_X = \frac{\rho_{ZX} \sigma_X \sigma_Z-  \rho_{ZW} \rho_{XW} \sigma_X \sigma_Z}{ \sigma_Z^2 -  \rho_{ZW}^2 \sigma_Z^2} \beta_X =  \frac{ \beta_X \sigma_X (\rho_{ZX} -\rho_{ZW} \rho_{XW} )}{\sigma_Z (1- \rho_{ZW}^2)}.$$ By Equation \eqref{eq:OmVarBias} and the above result, the omitted variable bias is 
$$\bbzOm = A^{-1} B \beta_X =   \frac{\beta_X \sigma_X \rho_{ZX}  }{\sigma_Z} =\bbzME  \frac{\rho_{ZX} (1-\rho_{ZW}^2)}{\rho_{ZX} -\rho_{ZW} \rho_{XW}} 
= \frac{(1-\rho_{ZW}^2) \cdot \bbzME}{1-\rho_{ZW}^2 ( \frac{ \rho_{XW}}{\rho_{ZX} \rho_{ZW} })}.$$ 
Since by assumption $ \rho_{XW} > \vert \rho_{ZW} \vert  > 0$ and $\text{sign} (\rho_{ZW}) = \text{sign}(\rho_{ZX})$ with $\vert \rho_{ZX} \vert \geq \vert \rho_{ZW} \vert$, 
$$ \frac{ \rho_{XW}}{\rho_{ZX} \rho_{ZW} }>1\quad \text{and so} \quad \rho_{ZW}^2  \frac{ \rho_{XW}}{\rho_{ZX} \rho_{ZW} }=\frac{\rho_{ZW} \rho_{XW}}{\rho_{ZX}} < 1.$$ These two inequalities imply that $$\frac{1-\rho_{ZW}^2 }{1-\rho_{ZW}^2 ( \frac{ \rho_{XW}}{\rho_{ZX} \rho_{ZW} })}  \in (1,\infty),$$
so by the above formula for $\bbzOm$ in terms of $\bbzME$, $\text{sign} (\bbzME) = \text{sign}(\bbzOm)$, and $\vert \bbzME \vert \leq \vert \bbzOm \vert$ with $\vert \bbzME \vert=\vert \bbzOm \vert$ if and only if $\bbzOm=0$.

\subsection{OVB as limiting case of classical MEB}\label{sec:LimitCME}

In the classical measurement error model $W=X+E$, where $E$ is a random vector of errors with mean $0$ that are independent of $X$, $Z$ and $Y$. Letting $\cov(E)=\Sigma_E$ we get
%the covariance matrix of $E$ be denoted by $\Sigma_E$, 
$C=\cov(Z,X+E)=B$, $F=\cov(X,X+E)=D$, and $G=\cov(X+E,X+E)=D+\Sigma_E$. By plugging these formulas into Equation \eqref{eq:MEbiasZ}, $$\begin{aligned}
    \bbzME %& =(A-C G^{-1} C^ \tran )^{-1} (B  -CG^{-1} F^ \tran )\beta_X \\ &
    = \big(A- B (D+\Sigma_E)^{-1} B^\tran  \big)^{-1} \big(B  -B (D+\Sigma_E)^{-1} D^ \tran \big) \beta_X.
\end{aligned}$$

Let $\lambda_{\text{min}}(\cdot)$ 
%and $\lambda_{\text{max}}(\cdot)$ 
denote the minimum  eigenvalue of an input matrix. We consider what happens to $\bbzME$ in the limit as the classical measurement error goes to infinity, in the sense that $\lambda_{\text{min}}(\Sigma_E) \to \infty$ while $A,B$, and $D$ remain fixed. 
In that limit $(D+\Sigma_E)^{-1}\to0$, and hence,
$$\bbzME \to A^{-1}B\beta_X=\bbzOm.$$
That is, the limiting value of the MEB is the OVB.

\subsection{If $d>p$, we cannot guarantee that OVB is worse } 

If there are more covariates measured with error than %covariates measured 
without error, the MEB is not necessarily worse than the OVB, even under a classical measurement error assumption. Recall formulas \eqref{eq:OmVarBias} and \eqref{eq:MEbiasZ} that $$\bbzOm = A^{-1} B \beta_X \quad \text{and} \quad \bbzME =\big( A- C G^{-1} C^ \tran \big)^{-1} (B - C G^{-1} F^ \tran ) \beta_X.$$ 
%Note that s
Since $B \in \mathbb{R}^{p \times d}$ with $p>d$, there exists some $\beta_X \neq 0$ in the null space of $B$. It is possible that $\beta_X$ will be in the null space of $B$ but not in the null space of $(B - C G^{-1} F^ \tran )$, in which case $\bbzOm=0$ but $\bbzME \neq 0$.

The scenario where $\bbzOm=0$ but $\bbzME \neq 0$ can even happen under the strong assumption that the measurement error is classical, where $W=X+E$, and $E$ is a random vector of errors with mean $0$ that are independent of $X$, $Z$ and $Y$. For example, suppose $p=1$, $d=2$, $\beta_X =(1,-2)$, $A=\var(Z)=1$, $B=\cov(Z,X)=\begin{bmatrix} 0.2 & 0.1 \end{bmatrix}$, 
$$D=\cov(X) =  \begin{bmatrix} 1 & 0.15 \\ 0.15 & 1 \end{bmatrix} \text{and } W=X+E \ \text{ where } \e[E]=0, \cov(E)= \begin{bmatrix} 0.7 & 0.05 \\ 0.05 & 0.4 \end{bmatrix},$$ 
and $E$ is independent of $X$, $Z$, and $Y$. In this setting, $F=\cov(X,W)=D$, $G=\cov(W)=D+\cov(E)$, so by formulas \eqref{eq:OmVarBias} and $\eqref{eq:MEbiasZ}$, $\bbzOm=0$ while  $\bbzME \approx 0.0257.$

\subsection{Counterexample under no restrictions on measurement error structure}

In this subsection, we show that if there are no restrictions on the structure of the measurement error, then it is possible to have no OVB while having arbitrarily large MEB. To see this we present a counterexample, although we believe a similar measurement error structure is unlikely to arise in practice.

Suppose $p=d=1$ and that $W=X+\kappa Z$, so $A$, $B$, $C$, $D$, $F$ and $G$ are scalars, with $C=B+\kappa A$, $F=D+\kappa B$, and $G=D+2\kappa B+\kappa^2 A.$ Plugging these values into \eqref{eq:MEbiasZ}, 
$$ \bbzME = \frac{B- \frac{(B+\kappa A)(D+\kappa B)}{D+2\kappa B+\kappa^2 A} }{A- \frac{(B+\kappa A)^2}{D+2\kappa B+\kappa^2 A}} \beta_X 
 =\frac{\kappa B^2-\kappa AD}{AD-B^2} \beta_X 
 =-\kappa \beta_X.$$

Meanwhile by \eqref{eq:OmVarBias}, $\bbzOm =A^{-1} B \beta_X$. So if we take $\kappa \to \infty$, but keep the distribution of $(Z,X)$ fixed, $\bbzOm$ will remain fixed but $\bbzME$ will go to either $\pm \infty$. If we consider the special case where there is no correlation between $Z$ and $X$, $B=0$ so  $\bbzOm =0$, while $\bbzME$ can be arbitrarily large as $\kappa \to \infty$.

\section{Proofs of Proposition \ref{prop:OmBiasNegConds} and Theorem \ref{thoerem:PollutantBiasOmegaSigns}}\label{sec:ProofOfTheorems}

\subsection*{Proof of Proposition \ref{prop:OmBiasNegConds}} Assume that the conditions of Proposition \ref{prop:OmBiasNegConds} hold and that without loss of generality we may suppose that $k=p$.
     For $j \in \{1,\ldots,d \}$ define $\bm{v}_j \equiv \text{Cov}(Z,X^{(j)} ) \in \mathbb{R}^p$ which is the jth column of $B$ and define $\bm{e}_p=(0,\ldots,0,1)$ to be the $p$th standard basis vector in $\mathbb{R}^p$. By Equation \eqref{eq:OmVarBias}, \begin{equation}\label{eq:OmVarBiasDecomp}
        [\bbzOm]_p = \bm{e}_p^\tran \bbzOm = \bm{e}_p^\tran A^{-1} B \beta_X = \sum_{j=1}^d \bm{e}_p^\tran A^{-1} \bm{v}_j \beta_{X,j}.
    \end{equation} We wish to show that $[\bbzOm]_p \leq 0$, and since by the No Benefit Assumption, $\beta_{X,j} \leq 0$ for all $j \in \{1,\ldots,d \}$, it will suffice to show that $\bm{e}_p^\tran A^{-1} \bm{v}_j \geq 0$ for all $j \in \{1,\ldots,d\}$.

    Fix $j \in \{1,\ldots,d\}$ and we will show that $\bm{e}_p^\tran A^{-1} \bm{v}_j \geq 0$. To do this define $M$ to be the covariance matrix of $(Z,X^{(j)})$ and note that by the blockwise inversion formula $$M = \begin{bmatrix} A & \bm{v}_j \\ \bm{v}_j^\tran & D_{jj} \end{bmatrix} \in \mathbb{R}^{(p+1)\times (p+1)} \Rightarrow \big[ M^{-1} \big]_{p,p+1} = -  \bm{e}_p^{\tran} (A- \bm{v}_j D_{jj}^{-1} \bm{v}_j^\tran )^{-1} \bm{v}_j D_{jj}^{-1}.$$ 
Since $\bm{v}_j$ is a vector and $D_{jj}$ is a scalar, note that by the Sherman-Morrison formula, $$\begin{aligned} 
\big[ M^{-1} \big]_{p,p+1}  &= 
  -  \frac{1}{D_{jj}} \bm{e}_p^{\tran} \Big(A^{-1} +  \frac{D_{jj}^{-1} A^{-1} \bm{v}_j \bm{v}_j^\tran A^{-1}}{1-D_{jj}^{-1} \bm{v}_j^{\tran} A^{-1} \bm{v}_j } \Big) \bm{v}_j \\
& = -  \frac{1}{D_{jj}} \bm{e}_p^{\tran} A^{-1} \bm{v}_j \Big( 1 + \frac{ \bm{v}_j^{\tran} A^{-1} \bm{v}_j}{D_{jj}- \bm{v}_j^{\tran} A^{-1} \bm{v}_j} \Big). 
\end{aligned}$$

We now express $\big[ M^{-1} \big]_{p,p+1}$ in terms of the partial correlation between $Z^{(p)}$ and $X^{(j)}$ conditional on $Z \setminus Z^{(k)}$. While partial correlations are defined in terms of the correlation between population residuals, we leverage the following well known formula for partial correlations (see, for example, \cite{PartialCorrFormulaInIntro}). Given a random vector $V=(V^{(1)},\ldots,V^{(m)})$ whose covariance matrix is $\Sigma$, the partial correlation 
\begin{align}\label{eq:partialcorr}
\corr(V^{(\ell)},V^{(r)} \mid V \setminus \{V^{(\ell)},V^{(r)} \} ) = -[\Sigma^{-1}]_{\ell r}/\sqrt{[\Sigma^{-1}]_{\ell \ell}[\Sigma^{-1}]_{rr}}
\end{align}
for $\ell \neq r$. Since $M$ is the covariance matrix of $(Z^{(1)},\ldots,Z^{(p-1)},Z^{(p)},X^{(j)})$, $$\big[ M^{-1} \big]_{p,p+1} = - \sqrt{\big[ M^{-1} \big]_{p,p} \big[ M^{-1} \big]_{p+1,p+1}} \corr \big( Z^{(p)},X^{(j)} \bigm| (Z^{(1)},\ldots,Z^{(p-1)}) \big) .$$
% \bigm| and \Bigm| put a touch of extra space vs \big| and \Big|

By the Weak $\PPC$ Assumption, the partial correlation in the above formula must be nonnegative. Hence $\big[ M^{-1} \big]_{p,p+1} \leq 0$. Combining this with a previous result, \begin{align}\label{eq:thingwithparens} -\frac{1}{D_{jj}} \Big( 1 + \frac{ \bm{v}_j^{\tran} A^{-1} \bm{v}_j}{D_{jj}- \bm{v}_j^{\tran} A^{-1} \bm{v}_j} \Big) \bm{e}_p^{\tran} A^{-1} \bm{v}_j = \big[ M^{-1} \big]_{p,p+1} \leq 0. \end{align}

By the positive definiteness in Condition (iii), $A^{-1} \succ 0$ and $M \succ 0$. Hence, $D_{jj}- \bm{v}_j^{\tran} A^{-1} \bm{v}_j>0$ because it is the Schur complement of $M$, and also $\bm{v}_j^{\tran} A^{-1} \bm{v}_j>0$. Thus the term in parentheses in Equation~\eqref{eq:thingwithparens} is strictly positive. Combining this fact with the above inequality and noting that $D_{jj}>0$, it follows that $\bm{e}_p^\tran A^{-1} \bm{v}_j \geq 0$. Since this argument holds for any $j \in \{1,\ldots,d\}$, combining Equation \eqref{eq:OmVarBiasDecomp} with the No Benefit Assumption completes the proof.  \qedsymbol \newline

In order to prove Theorem \ref{thoerem:PollutantBiasOmegaSigns}, we first introduce two special categories of matrices that we need.

\begin{definition}\label{def:zmmatrices}
A matrix $M \in \mathbb{R}^{m \times m}$ is a $Z$-matrix if $M_{\ell r} \leq 0$ for all $\ell,r \in \{1,\ldots,m\}$ with $\ell \neq r$. A $Z$-matrix $M \in \mathbb{R}^{m \times m}$ is an $M$-matrix 
if all of its real eigenvalues are strictly positive. 
\end{definition}
While there are over 40 equivalent definitions of an $M$-matrix \citep{Plemmons1977}, Definition~\ref{def:zmmatrices} is the most convenient for our purposes. The proof of the following theorem uses the Pairwise $\PPC$ Assumption and Condition (iii) to show that the lower right block of the precision matrix of $(Z,X)$ is an $M$-matrix and subsequently leverages the fact that the inverse of an $M$-Matrix has only nonnegative entries. 

\subsection*{Proof of Theorem \ref{thoerem:PollutantBiasOmegaSigns}}

Assume that the conditions of Theorem \ref{thoerem:PollutantBiasOmegaSigns} hold. Recall that the classical measurement assumption has $W=X+E$ where $E$ is independent of $(Z,X,Y)$. Therefore
$C=\cov(Z,W)=\cov(Z,X)=B\in\mathbb{R}^{p\times d}$,
%$B=\cov(Z,X)=\cov(Z,W)=C \in \mathbb{R}^{p \times d}$, 
$G=\cov(W,W)=\cov(X,X)+\cov(E,E)=D+\Sigma_E \in \mathbb{R}^{d \times d}$ and $F^\tran =\cov(W,X)=\cov(X+E,X)=D$. Hence by formula \eqref{eq:MEbiasAll}, the MEB is given by 
$$\bbME = \begin{bmatrix} A & C \\ C^\tran & D+\Sigma_E \end{bmatrix}^{-1} \begin{bmatrix} 0 & 0 \\ 0 & -\Sigma_E \end{bmatrix} \beta =  -  \begin{bmatrix} A & C \\ C^\tran & \Sigma_E +D \end{bmatrix}^{-1} \begin{bmatrix} 0 \\ \Sigma_E \beta_X \end{bmatrix}.$$ Thus by the blockwise matrix inversion formula and the above formula, Equation \eqref{eq:MEbiasAssumption} holds.
  
Now we will show that $\Omega_{jj} >0$ for all $j$. To see this note that $D- C^ \tran A^{-1} C \succ 0$ because it is the Schur complement of the covariance matrix of $(Z,X)$, which by Condition (iii) is positive definite. Since $\Sigma_E \succeq 0$, we also have $\Sigma_E +D- C^ \tran A^{-1} C \succ 0$ and then $\Omega = (\Sigma_E + D - C^\tran A^{-1} C)^{-1} \succ 0$. Hence $\Omega_{jj} > 0$ for all $j$.

By the CUME Assumption, $\Sigma_E\succeq0$.  First we
consider the case where $\Sigma_E \succ 0$, and will show that $\Omega_{jj'} \leq 0$ for $j \neq j'$ while $\Omega_{jj} < 1/a_j$, deferring consideration of $a_j=0$ to later.  With $\Sigma_E \succ 0$, $\Sigma_E$ is invertible so by the Woodbury Matrix Identity, 
\begin{equation}\label{eq:OmegaWhenSigmaEIsPD}
    \Omega = \Sigma_E^{-1} -\Sigma_E^{-1} \Big( (D-  C^\tran A^{-1} C)^{-1} +\Sigma_E^{-1}  \Big)^{-1} \Sigma_E^{-1}.
\end{equation} Let $\Sigma \in \mathbb{R}^{(p+d)\times (p+d)}$ be the covariance matrix of $(Z,X)$.
%$$\Sigma \equiv \cov \Big( \begin{bmatrix} Z \\ X \end{bmatrix},  \begin{bmatrix} Z \\ X \end{bmatrix} \Big) = \begin{bmatrix} A & C \\ C^\tran & D \end{bmatrix},$$ and 
By formula~\eqref{eq:partialcorr} for the partial correlations, $$[\Sigma^{-1}]_{p+j,p+j'} =-\sqrt{[\Sigma^{-1}]_{p+j,p+j}[\Sigma^{-1}]_{p+j',p+j'}} 
 \cdot \corr \big( X^{(j)},X^{(j')} \bigm| (Z,X) \setminus \{ X^{(j)},X^{(j')} \} \big),$$ for any $j, j' \in \{1,\ldots,d \}$ such that $j \neq j'$.
 It follows that by the Pairwise $\PPC$ Assumption, $[\Sigma^{-1}]_{p+j,p+j'} \leq 0$ when $1 \le j \neq j' \le d$. Since $\Sigma$ has blocks $A$, $B=C$, $B^\tran=C^\tran$, and $D$, by the blockwise inversion formula, the lower right block of $\Sigma^{-1}$ is $(D-  C^\tran A^{-1} C)^{-1}$. Hence the off-diagonal entries of $(D-  C^\tran A^{-1} C)^{-1}$ must be nonpositive, so $(D-  C^\tran A^{-1} C)^{-1}$ is a $Z$-matrix. Meanwhile, $(D-  C^\tran A^{-1} C)^{-1} \succ 0$ because as mentioned previously $D-  C^\tran A^{-1} C \succ 0$. The matrix $(D-  C^\tran A^{-1} C)^{-1}$ is also symmetric because it is the lower right block of $\Sigma^{-1}$ where $\Sigma$ is symmetric. Since $(D-  C^\tran A^{-1} C)^{-1} \succ 0$ is symmetric, all its eigenvalues are real and positive and so $(D-  C^\tran A^{-1} C)^{-1}$ is an $M$-matrix.

 Since $\Sigma_E^{-1}$ is a nonnegative diagonal matrix and $(D-  C^\tran A^{-1} C)^{-1}$ is an M-matrix, by Theorem 2.1 in \cite{Plemmons1977} the inverse of their sum must not have any negative entries. Thus  $\big( (D-  C^\tran A^{-1} C)^{-1} +\Sigma_E^{-1}  \big)^{-1}$ has only nonnegative entries. Because $\Sigma_E^{-1}$ and $\big( (D-  C^\tran A^{-1} C)^{-1} +\Sigma_E^{-1}  \big)^{-1}$ both only have nonnegative entries, their sandwich product, $\Sigma_E^{-1} \big( (D-  C^\tran A^{-1} C)^{-1} +\Sigma_E^{-1}  \big)^{-1} \Sigma_E^{-1}$, only has nonnegative entries. Because $\Sigma_E^{-1}$ is diagonal and by Equation \eqref{eq:OmegaWhenSigmaEIsPD}, it follows that  $\Omega_{jj'} \leq 0$ for all $j \neq j'$. 

 By Equation \eqref{eq:OmegaWhenSigmaEIsPD} and since $(D- C^ \tran A^{-1} C)^{-1} \succ 0$ and $\Sigma_E^{-1} \succ 0$, it must follow that $$\Sigma_E^{-1}-\Omega = \Sigma_E^{-1} \big( (D-  C^\tran A^{-1} C)^{-1} +\Sigma_E^{-1}  \big)^{-1} \Sigma_E^{-1} \succ 0.$$ Thus $[\Sigma_E^{-1}-\Omega]_{jj} >0$. Since $[\Sigma_E^{-1}-\Omega]_{jj} >0$ and $[\Sigma_E^{-1}]_{jj}=1/a_j$, $\Omega_{jj} < 1/a_j$. We have thus shown that in the case where $\Sigma_E \succ 0$, $\Omega_{jj'} \leq 0$ for $j \neq j'$ while $\Omega_{jj} < 1/a_j$.

Now consider the case where $\Sigma_E \nsucc 0$, so $a_{j^*}=0$ for one or more $j^* \in \{1,\ldots,d\}$, and clearly $\Omega_{j^*j^*} \le 1/a_{j^*}=\infty$. To upper bound $\Omega_{jj}$ and $\Omega_{jj'}$, consider a sequence of diagonal matrices $(\Sigma_E^{(m)})_{m=1}^{\infty}$ with strictly positive diagonal entries such that $\lim_{m \to \infty} \Sigma_E^{(m)} = \Sigma_E$. Observe that $g : \mathbb{R}^{d \times d} \to \mathbb{R}^{d \times d}$ given by $g(M)= (M +D-  C^\tran A^{-1} C )^{-1}$ is well defined at $\Sigma_E$ and therefore continuous at $\Sigma_E$. Thus, $\Omega=g(\Sigma_E) = \lim_{m \to \infty} g(\Sigma_E^{(m)})$, while by previous results for $j \neq j'$ and for all $m \in \mathbb{N}$, $[g(\Sigma_E^{(m)})]_{jj'} \leq 0$, $[g(\Sigma_E^{(m)})]_{jj} < 1/a_j$. Hence for $j \neq j'$, $$\Omega_{jj'} = \lim\limits_{m \to \infty} \Big[g(\Sigma_E^{(m)}) \Big]_{jj'} \leq 0 \quad \text{and} \quad \Omega_{jj} = \lim\limits_{m \to \infty} \Big[g(\Sigma_E^{(m)}) \Big]_{jj} \leq 1/a_j.$$ Thus we have shown $\Omega_{jj'} \leq 0$ and $\Omega_{jj} \leq 1/a_j$ for any $j \neq j'$. \qedsymbol

\section{Details on random covariance matrices in our simulations}\label{sec:MCSimsRandomCovMatricesDetails}

In each Monte Carlo simulation we generated a random covariance matrix for $(Z,X,W)$ in order to meet a few desiderata. First, we wanted the sample of random covariance matrices to cover a broad range of possible correlation structures for $(Z,X)$. Second, we wanted the sample of random covariance matrices to include sufficiently many scenarios where the average pairwise correlation between the pollutants was positive, as we expect that in most motivating applications the average pairwise correlation of the pollutants is positive. Third, we wanted the measurement error to follow a more general structure than that assumed in Section \ref{sec:MEbias_genConditions} to see if the conclusions of Corollaries \ref{cor:NegBiasNoME} and \ref{cor:NegBiasNulPol} tended to hold without the restrictive CUME assumption.

To meet the third desideratum, we supposed that $X$ and $W$ came from the Berkson-Classical mixture model described in Section 3.2.5.5 of \cite{CarrolStefanski06}, because the model is flexible and easy to work with in simulations. \footnote{Another rational for using a Berkson-Classical mixture model is that there are notions in air pollutant measurement error literature that the measurement error can be decomposed into a sum of a Berkson-like component and a classical-like component \citep{Szpiro2011,SzpiroPaciorekEnvironmetrics,SpatialSimex}. However, as noted in Section 3.2 of \cite{Szpiro2011}  the classical-like component is distinct from classical measurement error in important ways, so, in spite of the aforementioned decomposition, the Berkson-Classical measurement error model we use does not capture all measurement error structures of interest.} In particular we supposed that $X= \mathcal{L}_X +U_b$ and $W= \mathcal{L}_X +U_c$, where $\mathcal{L}_X$, $U_b$, and $U_c$ are 3 independent mean $0$ random vectors in $\mathbb{R}^5$ which are independent of $Z$ and denote the latent vector, the Berkson-type error and the Classical-type error component, respectively.  Defining $\Sigma_b \equiv \text{Cov}(U_b)$ and $\Sigma_c \equiv \text{Cov}(U_b)$, in our Monte-Carlo simulations we drew 
$$\Sigma_b,\Sigma_c \stackrel{\mathrm{ind}}{\sim}\text{Wishart}_5 \Big(\frac{1}{3} I_5, 10 \Big).$$ 
%\quad \text{ and independently } \Sigma_c \sim \text{Wishart}_5 \Big( \frac{1}{3} I_5, 10 \Big).$$ 
Now the covariance matrix of $W-X$ is $\Sigma_b +\Sigma_c$ which is clearly not diagonal, and since the number of degrees of freedom $10$ is not much larger than $5$, $\Sigma_b +\Sigma_c$ will often be far from being diagonal (even after appropriately scaling to have diagonal entries that average $1$). This ensures that, unlike our theoretical analysis in Section \ref{sec:MEbias_genConditions}, our simulations explore situations where the measurement errors across pollutants are quite correlated.

Next we generated a random covariance matrix for $(Z,\mathcal{L}_X)$, by drawing $$\Sigma_{\mathcal{L}} \sim \text{Wishart}_{10} ( V_{\mu} , 10 ), \quad \text{where }\  V_{\mu} \equiv \begin{bmatrix} I_4\  & 0 \\[1ex] 0\  & \dfrac45 I_6 + \dfrac15 \bm{1}_6 \bm{1}_6^\tran  \end{bmatrix} .$$ %Above $0.8 I_6 + 0.2 \bm{1}_6 \bm{1}_6^\tran$ is a $6 \times 6$ matrix with 1 on the diagonals and $0.2$ of the diagonals, while $I_4$ is the $4 \times 4$ identity matrix. 
Note that $\Sigma_{\mathcal{L}}$ is drawn from a Wishart distribution with the minimum allowable number of degrees of freedom, given the $10 \times 10$ dimension of $\Sigma_{\mathcal{L}}$. This means we are drawing $\Sigma_{\mathcal{L}}$ from a distribution with the largest possible spread (among Wishart distributions with scale matrix $V_{\mu}$) across the space of positive semi-definite $10 \times 10$ matrices. This choice of degrees of freedom allows us to meet our first desideratum that we wanted to cover a broad range of possible correlation structures for $(Z,X)$. Also note that we chose the off-diagonal entries corresponding to pairs of pollutants in $V_{\mu}$ to be a positive value of $0.2$ in order to meet our second desideratum that the simulations would include many settings where the pollutants had positive pairwise correlations.

After generating $\Sigma_b$, $\Sigma_u$, and $\Sigma_{\mathcal{L}}$ the covariance matrix for $(Z,X,W)$ was given by, $$\cov \Bigg( \begin{bmatrix} Z \\ X  \\ \tx \end{bmatrix}  \Bigg) = \begin{bmatrix} A & B & C \\ B^\tran & D & F  \\ C^\tran & F^\tran & G \end{bmatrix} = \begin{bmatrix} [\Sigma_{\mathcal{L}}]_{1:5,1:5} & [\Sigma_{\mathcal{L}}]_{1:5,6:10} &  [\Sigma_{\mathcal{L}}]_{1:5,6:10} \\ [\Sigma_{\mathcal{L}}]_{6:10,1:5} & \quad[\Sigma_{\mathcal{L}}]_{6:10,6:10} + \Sigma_b & [\Sigma_{\mathcal{L}}]_{6:10,6:10}  \\  [\Sigma_{\mathcal{L}}]_{6:10,1:5} & [\Sigma_{\mathcal{L}}]_{6:10,6:10} &  [\Sigma_{\mathcal{L}}]_{6:10,6:10} + \Sigma_c \end{bmatrix}.$$

\begin{figure}[t]
\centering
\includegraphics[width=0.95 \hsize]{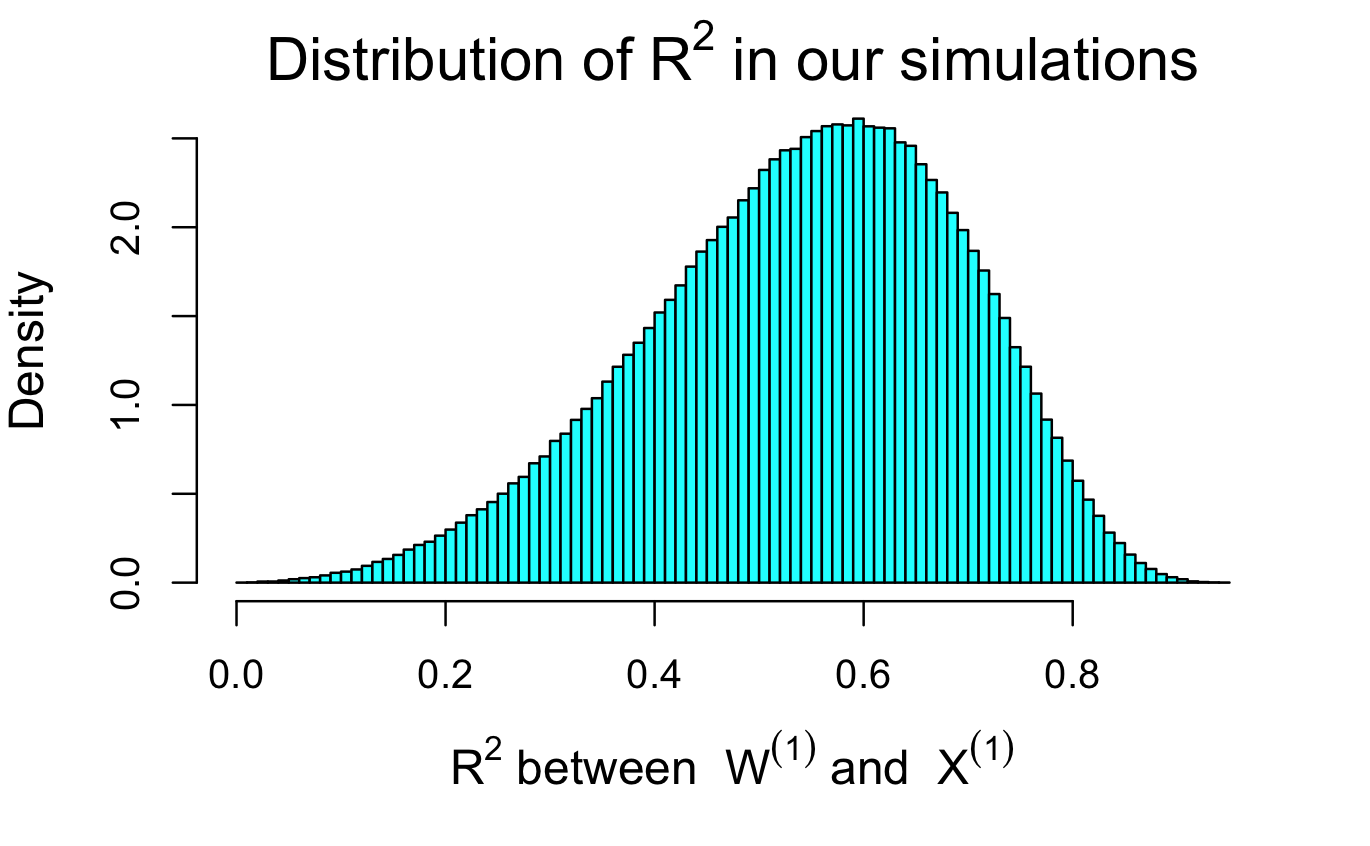}
\caption{\label{fig:R2Dist} Distribution of $R^2$ values between $W^{(1)}$ and $X^{(1)}$ in our simulations. For $j=2,\ldots,5$ this histogram also gives the distribution of the $R^2$ values between $W^{(j)}$ and $X^{(j)}$ because of the symmetry in generating $W$ and $X$ in our simulations.}
\end{figure}

We remark that the factor of $1/3$ in the first parameter of the Wishart distribution from which $\Sigma_b$ and $\Sigma_c$ were sampled was set so that 
%the 
$W$ would be a predictor of $X$ that has realistic accuracy. In particular, if we chose a factor lower than $1/3$, $W$ would tend to be a more accurate predictor of $X$ and if we chose a higher factor than $1/3$, $W$ would tend to be a less accurate predictor of $X$. For $j=1,\ldots,5$, Figure \ref{fig:R2Dist} depicts the distribution of $R^2$ values between the error-prone proxy $W^{(j)}$ and the true value $X^{(j)}$ across our simulations. Many data products for air pollutants have $R^2$ with the ground truth in the $0.3$ to $0.85$ range and our simulations reflect this.

\section{Rescaling coefficients in terms of WHO targets}\label{sec:WHO_rescaling}

To convert $\CO$, $\Ozone$, $\NitrogenDi$ and $\SulfurDi$ to be on the same scale as the WHO 2021 targets, we converted our concentration data of these pollutants from units of ppb to units of $\mu$g/m$^3$ based on the conversion factors at 20 °C and 101.3 kPa presented in Chapter 3 of \cite{WHO_guidelines2021}. The distribution of pollutant concentrations in our panel dataset after rescaling to be in units $\mu$g/m$^3$ can be found in Figure \ref{fig:AQSvsWHOtargets}. We subsequently rescale the pollutants based on the WHO 2021 targets (depicted as vertical dashed red lines in Figure \ref{fig:AQSvsWHOtargets}) based on the values presented in Table 3.26 of \cite{WHO_guidelines2021}. In particular, we divide the average June--August concentrations of $\CO$, $\NitrogenDi$, $\SulfurDi$, $\PMBig$, and $\PMSmall$ by $4{,}000$ $\mu$g/m$^3$, $25$ $\mu$g/m$^3$, $40$ $\mu$g/m$^3$, $45$ $\mu$g/m$^3$, and $15$ $\mu$g/m$^3$, respectively, which were the daily WHO targets for these pollutants set in 2021. We divide by daily targets rather than annual targets, because annual targets were not available for each of the pollutants and because none of the targets were specifically for the June--August period (see Table 3.26 of \cite{WHO_guidelines2021}). For $\Ozone$ we divided the June--August average of the daily maximum average concentration over each 8 hour window by $100$  $\mu$g/m$^3$, which was the WHO $\Ozone$ target for average concentration in each 8 hour window (in Table 3.26 of \cite{WHO_guidelines2021}, the WHO did not set targets for daily averages for $\Ozone$ as it did for the other 5 pollutants).

\begin{figure}[t]
\centering
\includegraphics[width=0.95 \hsize]{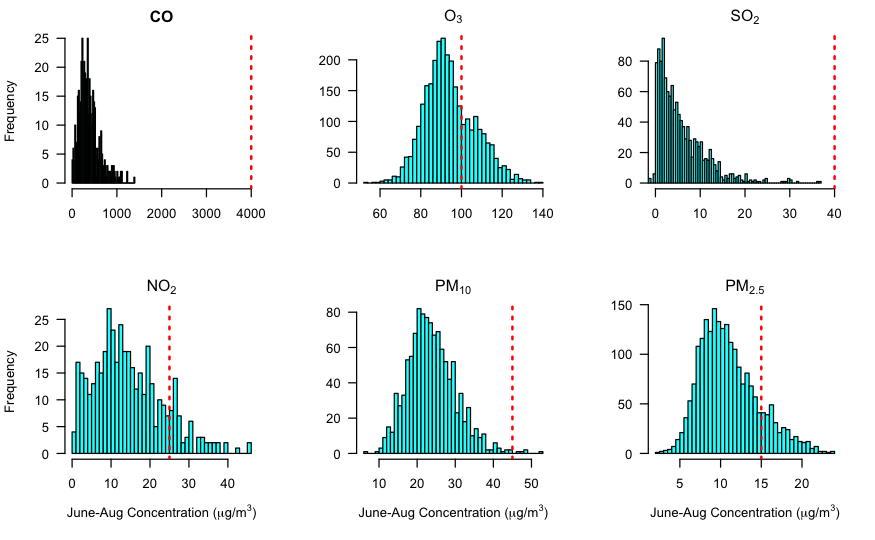}
\caption{\label{fig:AQSvsWHOtargets} Pollutant concentrations in dataset versus WHO 2021 targets. The histograms give the distribution of average, county-level June--August concentrations in our panel dataset. The dashed vertical line gives the daily target based on the WHO 2021 guidelines.}
\end{figure}

\end{document}